\documentclass[journal]{IEEEtran}
\usepackage{cite}
\usepackage{amsmath,amssymb,amsfonts,amsthm}
\usepackage{algorithmic}
\usepackage{subcaption}
\usepackage{graphicx}
\usepackage{textcomp}
\usepackage{xcolor}
\usepackage{enumitem}
\usepackage{tabularx,array}

\newtheorem{proposition}{Proposition}
\usepackage{lipsum}
\usepackage{flushend}
\usepackage{booktabs,makecell,multirow}

\newif\ifshowchanges
\showchangesfalse

\definecolor{rewblue}{RGB}{0,0,200}

\ifshowchanges
  \newcommand{\rew}[1]{\textcolor{rewblue}{#1}}
\else
  \newcommand{\rew}[1]{#1}
\fi

\title{Near-field Boundary Distance in mmWave and THz Communications with Misaligned Antenna Arrays}

\author{Peng Zhang,~\IEEEmembership{Graduate Student Member,~IEEE,} Vitaly Petrov,~\IEEEmembership{Member,~IEEE,}
        and~Emil Bj{\"o}rnson,~\IEEEmembership{Fellow,~IEEE}\vspace{-5mm}
\thanks{Peng Zhang, Vitaly Petrov, and Emil Bj{\"o}rnson are with KTH Royal Institute of Technology, Stockholm, Sweden. This work has been supported in part by Digital Futures at KTH, Grant 2022–04222 from the Swedish Research Council, Vinnova SweWIN Grant 2023–00572, and the Swedish Foundation for Strategic Research FFL–9 program. A shorter version of this work has been presented at the IEEE GLOBECOM 2025~\cite{zhang2025impact}.}}

\begin{document}
\maketitle

\begin{abstract}
Wireless communications in the millimeter wave (mmWave) and terahertz (THz) spectrum allow harnessing large frequency bands, thus achieving ultra-high data rates. However, the inherently short wavelengths of mmWave and THz signals lead to an extended radiative near-field region, where certain canonical far-field assumptions fail. Most prior works aimed to characterize this radiative near-field region either do not consider antenna arrays on \emph{both} communicating nodes or, if they do, assume perfect alignment between the arrays. However, such assumptions break down in many realistic deployments, where both sides must employ large-scale mmWave/THz antenna arrays to maintain the desired communication range, while perfect antenna alignment cannot be guaranteed particularly under nodes mobility. In this work, a generalized mathematical framework is presented to characterize the radiative near-field distance in directional mmWave and THz communication systems \emph{under various realistic array rotations and misalignments}. With the use of the developed framework, \emph{compact closed-form expressions} are derived for the near-field boundary distance in a wide range of antenna configurations, including array-to-array and array-to-point setups, considering both linear and planar arrays. Our numerical study reveals that the presence of antenna misalignment may significantly adjust the boundaries of the near-field region in mmWave and THz communication systems.
\end{abstract}


\section{Introduction}
Despite a moderate commercial penetration of millimeter wave (mmWave, $30$\,GHz--$300$\,GHz) communications in 5G networks~\cite{narayanan2022comparative}, both mmWave and even wider Terahertz frequency bands (THz, $300$\,GHz to $3$\,THz) are actively explored for ultra-high-throughput wireless communications beyond 5G~\cite{akyildiz2022terahertz,shafie2022terahertz,jiang2024terahertz}. However, extremely short wavelengths of mmWave and, especially, THz signals together with the envisioned deployment of cm-scale antenna arrays to maintain sufficient communication range, significantly expand the so-called ``radiative near-field region"~\cite{cui2022near}. This region may extend to several tens or even hundreds of meters~\cite{proc_ieee_jornet}, where conventional far-field approximations (such as the plane-wave assumption~\cite{pizzo2022fourier}) are not accurate~\cite{balanis2015antenna,singh2023wavefront}. Furthermore, even a slight antenna misalignment may lead to substantial variations in the actual ``\textit{near-field distance}" -- \textit{the distance \rew{beyond} which far-field approximations become sufficiently accurate}~\cite{monemi2024study}. Such variations further complicate the design and performance prediction. Consequently, accurate determination of the boundary between the radiative near-field and far-field regions is essential for the efficient design and optimization of the mmWave and THz communication systems of tomorrow.

\subsection{Prior Work on Near-field Distance Characterization}
The characterization of the radiative near-field distance in wireless communications has traditionally been addressed from the perspective of phase differences~\cite{balanis2015antenna,selvan2017fraunhofer}. The canonical boundary separating the radiative near-field from the far-field is defined by the Fraunhofer criterion, which defines the near-field distance as the threshold separation distance between the transmitter and the receiver antennas, where the maximal phase difference between the components of the signal received by different parts of the receiver antenna or antenna array remains below $\pi/8$~\cite{balanis2015antenna,selvan2017fraunhofer}. Initial derivations considered only a simple point-source transmitter (Tx) to a fixed-aperture receiver (Rx) scenario~\cite{stutzman2012antenna}. Recently, several studies (e.g., \cite{lu2023near,petrov2023near,renwang2025applicable}, among other relevant works) extended this analysis to more complex configurations involving uniform linear arrays (ULAs) and uniform planar arrays (UPAs).

Beyond phase-based characterizations, near-field boundaries have also been defined via received-power uniformity, multiplexing/capacity metrics, and beamfocusing-related criteria; see, e.g., \cite{lu2021does,lu2021communicating,cui2024near,bohagen2009spherical,wang2014tens,jiang2005spherical,bjornson2021primer,zeng2025revisiting,daei2025near} and the surveys \cite{liu2025near,liu2024near}.

\subsection{Near-field Distance with Misaligned Antennas}
In summary, multiple prior studies in the field present useful insights on the near-field distance characterization for different threshold criteria applied and antenna configurations. However, the overwhelming majority of such studies (including but not limited to~\cite{lu2023near,bjornson2021primer,bohagen2009spherical}), have one key simplistic assumption in common: \emph{they all explicitly or implicitly assume perfect alignment between the Tx and the Rx antenna systems}. Meanwhile, in many practical deployments (e.g., those involving mobile users), array misalignment is unavoidable due to node mobility or rotations. Furthermore, any misalignment between the Tx and Rx antennas notably affects the phase-related characteristics at the Rx~\cite{jayasankar2024impact}, hence challenging the applicability of the near-field distance formulations above that only consider a perfect antenna alignment~setup.

\emph{To the best of the authors' knowledge, only a few recent works have partially addressed angular variability or off-boresight antenna configurations.} Among those,~\cite{lu2021communicating} introduces a direction-dependent Rayleigh surface to capture angular effects, but the closed-form expressions are only valid for special cases, and, importantly, the array is only modeled at the base station side, while a mobile user is modeled as a point-source node. \rew{Similarly,~\cite{monemi2024revisiting} and its extended version~\cite{monemi2024study} examine the Fraunhofer distance in off-boresight scenarios, yet the analysis is restricted to simplified ULA-to-point configurations. In these works, the off-boresight boundary is defined via a worst-case element-to-reference (or inter-element) phase spread across the aperture, where the reference point may depend on the look direction. Along a related line,~\cite{hu2023design} studies near-field beamforming for large intelligent surfaces under a surface-to-point setup and adopts a similar phase-spread based criterion to characterize the near-field regime.} From different perspectives,~\cite{renwang2025applicable} derives applicable regions of spherical- and plane-wave models from channel-rank and gain perspectives, but the analysis again reduces to array-to-point cases. Likewise,~\cite{li2025near} investigates arc-shaped ultra-large arrays and characterizes their direction-dependent Rayleigh and uniform-power distances, yet the user side is still modeled as a single-antenna point source. Finally, our preliminary work on this topic~\cite{zhang2025impact} only studies the possible rotations of the antenna array on one of the communicating nodes.

\subsection{Our Novelty and Contribution}
Following the related work summary above, there exists a major gap in the literature: \emph{how to characterize the near-field distance in mmWave/THz communications using ULA or UPA systems \textbf{under realistic antenna array misalignment}}. Such a boundary provides a fundamental basis for determining the separation distance region, where far-field assumptions remain valid. Understanding the behavior and dynamics of the near-field distance under practical misalignment is also important for developing efficient engineering solutions for next-generation mmWave/THz systems (e.g., beamforming, channel estimation, and alignment strategies, among others).

To address the discussed gap, in this paper, we develop a generalized framework for characterizing the near-field distance under misalignment in array-to-array wireless links. The main contributions are summarized as follows:

\begin{itemize}
    \item \rew{Following the ongoing discussions in the community on the accuracy of using the canonical $\pi/8$ maximal phase deviation as a near-field to far-field threshold~\cite{stutzman2012antenna,balanis2015antenna,selvan2017fraunhofer}, \emph{we establish a generalized phase-error-based boundary criterion} with a tunable threshold $\varphi \in (0,\pi]$ (including $\varphi=\pi/8$ as a special case). Under this criterion, \emph{we develop a unified geometry-driven general approach} that characterizes the effective phase of each Tx--Rx component by jointly accounting for the propagation phase and the Tx/Rx beamsteering (phase compensation) terms under misalignment.}

    \item \rew{Building on the above generalized criterion and geometry-driven framework, \emph{we develop misalignment-aware closed-form radiative near-field distance expressions} for multiple canonical configurations, including array-to-array links (ULA-to-ULA and UPA-to-UPA) and representative array-to-point special cases. Moreover, to further illustrate the extensibility of the proposed approach beyond same-type arrays, we also provide \emph{representative mixed UPA-to-ULA examples} (e.g., the no-rotation case and the case with a single rotation angle), which demonstrate that the same geometry-driven analysis can be readily extended to arbitrary array-type combinations and rotation models by specifying the corresponding element coordinates and orientations.}
        
    \item \emph{We validate the developed expressions via extensive numerical simulations} and provide a quantitative analysis of the impact of misalignment on near-field boundaries, including the deformation of dominance regions and the scaling behavior with respect to array~aperture.
    
    \item \rew{\emph{We show that array misalignment can induce non-negligible deviations in the near-field boundary distance relative to conventional no-rotation expressions (see, e.g., Figs.~\ref{fig:diff_RxSize} and~\ref{fig:nearfield_distribution_alpha}).} Such deviations can be substantial in both relative and absolute terms (e.g., several tens of meters in practical mmWave/THz setups), and may therefore change the near-/far-field regime classification and the corresponding signal model/processing design.}
\end{itemize}

The rest of the paper is organized as follows. Section~\ref{sec:system_model} introduces the system model. Section~\ref{sec:ULA} presents the near-field distance derivation for misaligned ULA-to-ULA links. Section~\ref{sec:UPA} extends the analysis to UPA-based configurations. Section~\ref{sec:results} provides numerical results that validate the theoretical expressions and illustrate the impact of misalignment. Section~\ref{sec:conclusion} concludes the paper. Key notation is in Table~\ref{tab:notation}.

\begin{table}[!b]
\centering
\caption{List of Notations}
\label{tab:notation}
\renewcommand{\arraystretch}{1.2}
\begin{tabular}{>{\centering\arraybackslash}m{1.2cm}|m{6.5cm}}
\hline\hline
\textbf{Notation} & \textbf{Description}\\
\hline
\multicolumn{2}{c}{\textit{\textbf{Rotation-related parameters}}} \\
\hline
$\theta$ & Vertical rotation angle of the Tx \\
\hline
$\phi$ & Horizontal rotation angle of the Tx \\
\hline
$\alpha$ & Vertical off-boresight angle \\
\hline
\multicolumn{2}{c}{\textit{\textbf{Antenna-related parameters}}} \\
\hline
$\lambda$ & Wavelength of the carrier signal \\
\hline
$D_1, D_2$ & Size of the ULA or side length of the UPA at Tx/Rx \\
\hline
$N_1, N_2$ & Number of antennas in the Tx/Rx ULA, or side length of the Tx/Rx UPA, respectively \\
\hline
$r$ & Distance between the centers of the Tx and Rx arrays \\
\hline
\multicolumn{2}{c}{\textit{\textbf{ULA parameters}}} \\
\hline
$d_n$ & Displacement of the $n$-th element from the center of either Tx or Rx array \\
\hline
$n_1$, $n_2$ & Indices of elements on the TxULA and RxULA \\
\hline
$r_{n_2,n_1}$ & Distance between the $n_1$-th Tx element and the $n_2$-th Rx element \\
\hline
$r'_{n_2,n_1}$ & ULA effective distance \\
\hline
\multicolumn{2}{c}{\textit{\textbf{UPA parameters}}} \\
\hline
$m$, $n$, $i$, $j$ & Indices of elements on the UPA: $i$, $j$ for the TxUPA (horizontal/vertical), and $m$, $n$ for the RxUPA \\
\hline
$\mathbf{d}_{m,n}$ & Position vector of the $(m,n)$-th element on the RxUPA \\
\hline
$\mathbf{d}_{i,j}$ & Position vector of the $(i,j)$-th element on the TxUPA \\
\hline
$r^{i,j}_{m,n}$ & Distance between the $(m,n)$-th element on the RxUPA and the $(i,j)$-th element on the TxUPA \\
\hline
$r'^{i,j}_{m,n}$ & UPA effective distance \\
\hline\hline
\end{tabular}
\end{table}

\begin{figure}[!t]
    \centering
    \includegraphics[width=0.9\linewidth]{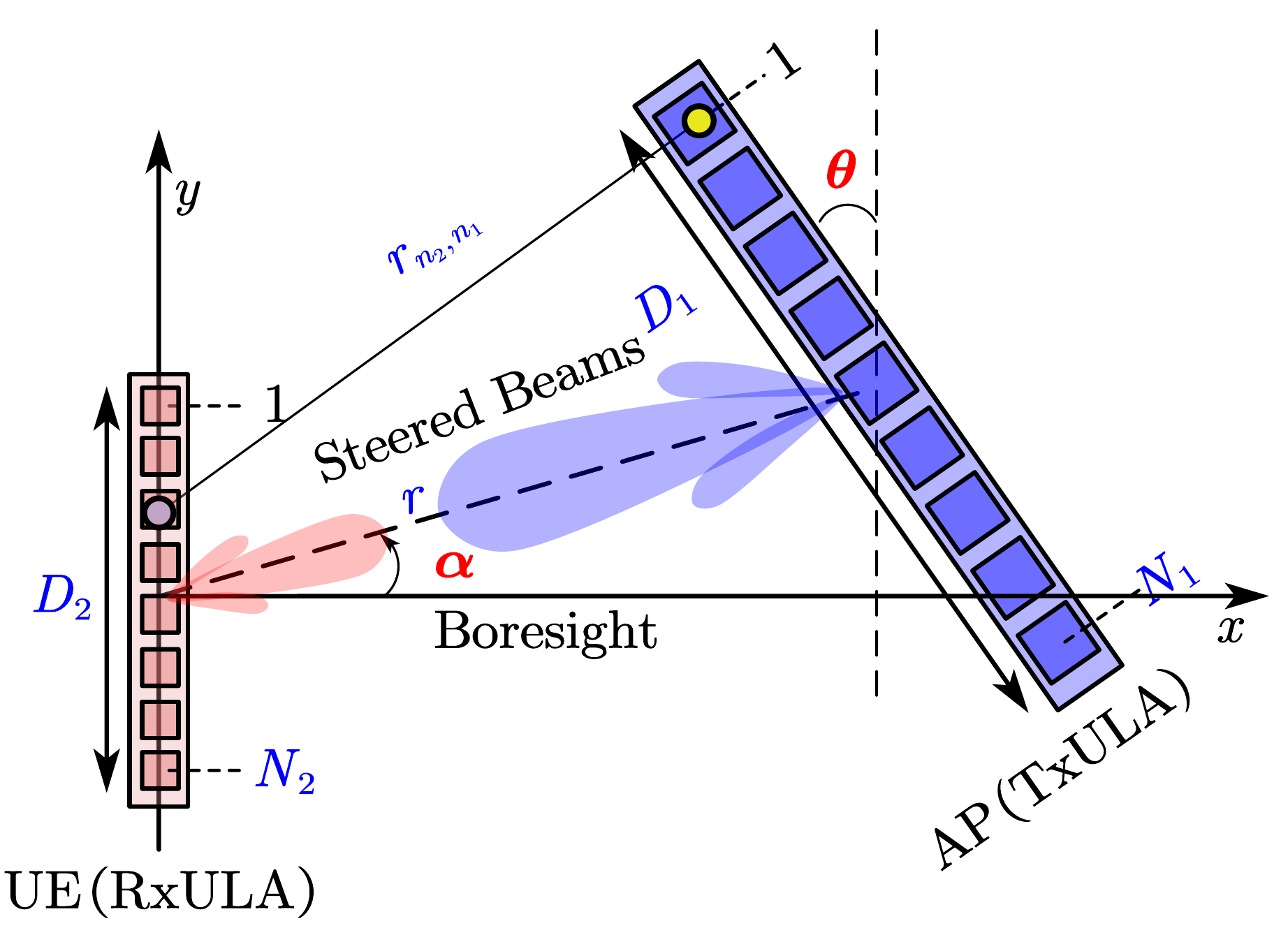}
    \caption{Considered ULA-to-ULA scenario (analyzed in Sec.~\ref{sec:ULA}). Section~\ref{sec:ULA-to-ULA(on-boresight)} covers the on-boresight case ($\alpha = 0$), while Sec.~\ref{sec:ULA-to-ULA(off-boresight)} -- the off-boresight case ($\alpha \neq 0$); Sec.~\ref{sec:ULA-to-point} explores a special ULA-to-point setup ($N_{1} = 1$).}
    \label{fig:ULA-to-ULA_system_model(off-boresight)}
\end{figure}

\begin{figure}[!t]
    \centering
    \includegraphics[width=0.9\linewidth]{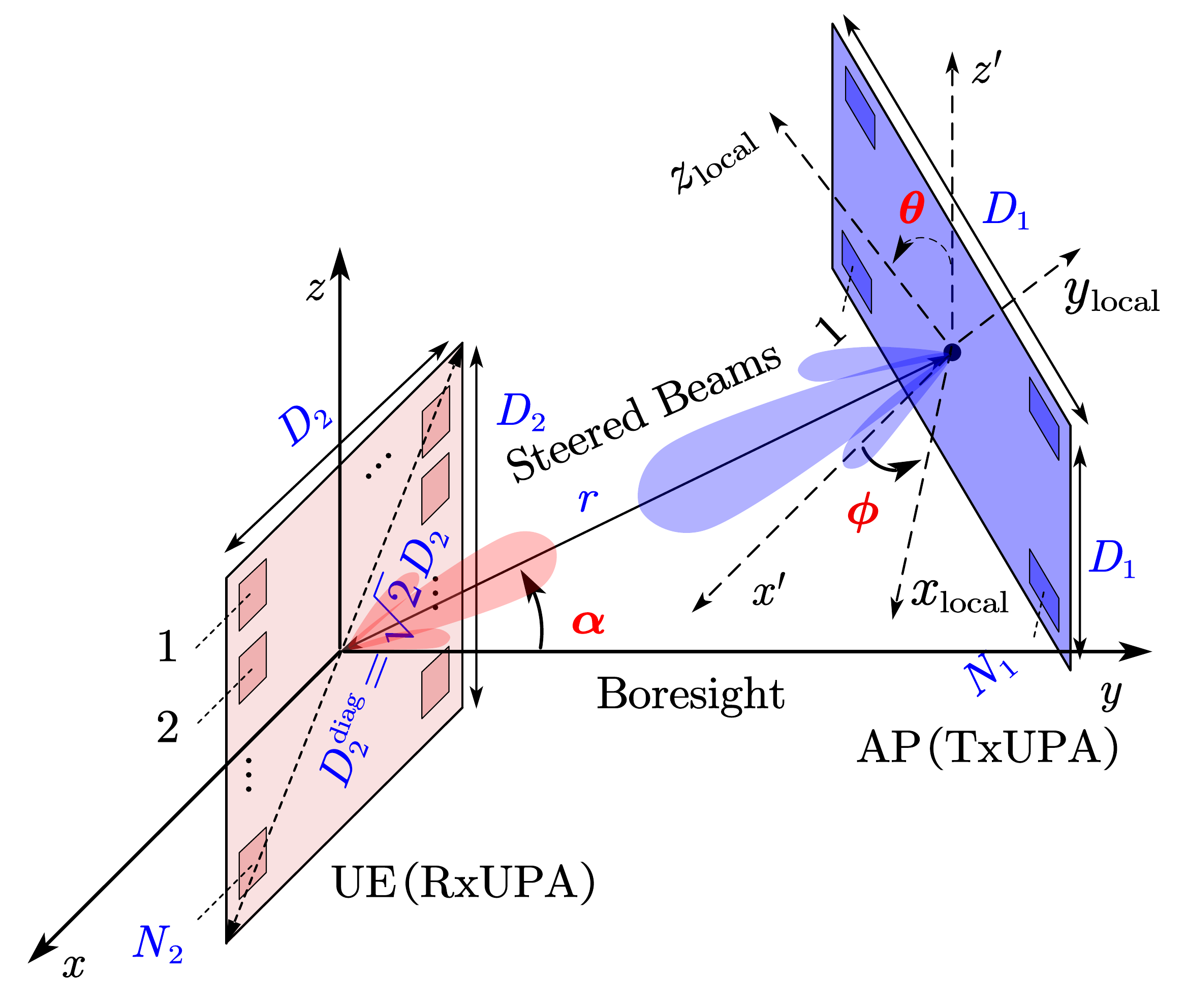}
    \caption{Considered UPA-to-UPA scenario (analyzed in Sec.~\ref{sec:UPA}). Section~\ref{sec:UPA-to-UPA(on-boresight)} covers the on-boresight case ($\alpha = 0$), Sec.~\ref{sec:UPA-to-UPA(off-boresight,1)} discusses the off-boresight case ($\alpha \neq 0$) with the Tx rotation in only one plane ($\phi = 0$, $\theta \neq 0$), while Sec.~\ref{sec:UPA-to-UPA(off-boresight,2)} -- the off-boresight case ($\alpha \neq 0$) with the Tx rotation in two planes ($\phi \neq 0$, $\theta \neq 0$); Sec.~\ref{sec:UPA-to-point} explores a special UPA-to-point setup ($N_{1} = 1$).}
    \label{fig:UPA-to-UPA_system_model(off-boresight,two_angles)}
\end{figure}

\section{System Model}\label{sec:system_model}
For clarity, we describe a downlink communication system, where the access point (AP) serves as the Tx and the user equipment (UE) is the Rx. Notably, the final expressions and results are also applicable to corresponding uplink scenarios (UE Tx and AP Rx) due to electromagnetic reciprocity~\cite{neiman2006principle}.
\subsection{ULA-to-ULA Case}
We first analyze an off-boresight system comprising a ULA-equipped AP (referred to as TxULA) and a ULA-equipped UE (referred to as RxULA), as in Fig.~\ref{fig:ULA-to-ULA_system_model(off-boresight)}. \rew{Both ULA arrays are assumed to lie in the same plane and the array misalignment is restricted to in-plane rotation.} The TxULA consists of $N_1$ antenna elements, while the RxULA contains $N_2$ elements, both with inter-element spacing of $\lambda/2$. Accordingly, the physical aperture lengths of the TxULA and RxULA are denoted by $D_1$ and $D_2$, respectively, with $D_1=(N_1-1)\lambda/2$ and $D_2=(N_2-1)\lambda/2$. \rew{The separation distance between the center of the TxULA and the center of the RxULA is set to $r$ meters. To avoid degenerate extremely-close geometries (e.g., overlapping apertures) we assume $r \geq (D_1+D_2)/2$.}

There are two key antenna rotation angles to consider in such a setup (see Fig.~\ref{fig:ULA-to-ULA_system_model(off-boresight)}). The first angle is an off-boresight angle $\alpha \in [-\pi/2, \pi/2]$, which denotes the angle between the line connecting the centers of the two arrays and the positive $x$-axis. A positive $\alpha$ indicates that the TxULA is positioned in the upper half-plane relative to the RxULA, while the negative $\alpha$ puts the TxULA to the lower half-plane. The second angle is the additional TxULA rotation angle $\theta$ around the $z$-axis. \rew{Without loss of generality, the Tx is assumed to be located on the positive $x$-axis. The case where the Tx is located on the negative $x$-axis is fully covered by the proposed analysis due to the inherent symmetry. In particular, it is geometrically equivalent to a mirror reflection of the considered configuration with respect to the $y$-axis (which can be represented by the angle mapping $\alpha \mapsto \pi-\alpha$ and $\theta \mapsto -\theta$).}

\subsection{UPA-to-UPA Case}
We then analyze the UPA-to-UPA configuration, as in Fig.~\ref{fig:UPA-to-UPA_system_model(off-boresight,two_angles)}, where both the AP and UE are equipped with UPAs. The Tx and Rx nodes are referred to as TxUPA and RxUPA, respectively. Similar to the ULA case, both arrays adopt an inter-element spacing of $\lambda/2$ in each dimension. The RxUPA consists of $N_2 \times N_2$ antenna elements and is fixed in the $xz$-plane, with its boresight aligned along the positive $y$-axis. The TxUPA, comprising $N_1 \times N_1$ elements, is initially positioned along the $y$-axis at a distance $r$ from the~UE. \rew{To avoid degenerate extremely-close geometries, we similarly assume $r\geq (D_1+D_2)/\sqrt{2} $, where $D_1$ and $D_2$ denote the side lengths of the TxUPA and RxUPA apertures, respectively, with $D_1=(N_1-1)\lambda/2$ and $D_2=(N_2-1)\lambda/2$.} In this configuration, the line connecting the centers of the TxUPA and RxUPA forms an angle $\alpha \in [-\pi/2, \pi/2]$ with respect to the $y$-axis, referred to as the off-boresight angle. A positive $\alpha$ indicates that the TxUPA is located above the boresight plane (i.e., in the upper $xy$ half-space) relative to the RxUPA. 

\rew{To model realistic misalignments due to small-scale motion and general orientation~\cite{petrov_small_scale}, the TxUPA is allowed to undergo a composite 3D rotation from its reference alignment.
To define the rotations, we introduce a translated global coordinate system $(x',y,z')$ whose origin is at the TxUPA center and whose axes are parallel to the global axes $(x,y,z)$ (see Fig.~\ref{fig:UPA-to-UPA_system_model(off-boresight,two_angles)}).
The TxUPA orientation is then specified by two angles: a vertical rotation $\theta$ around the $x'$-axis and a horizontal rotation $\phi$ around the $z'$-axis, following the \textit{right-hand rule}\footnote{\rew{The misalignment angles in this paper are treated as input parameters describing the relative Tx/Rx geometry and orientation. Depending on the deployment, they may be known from the installation geometry, available/estimated as side information (e.g., from device orientation sensing or alignment procedures), or uncertain. In the latter case, the same closed-form expressions can be used for statistical/robust analysis by modeling the angles as random variables, as illustrated in Section~V-D.}}.
The resulting rotated axes define the local coordinate system $(x_{\mathrm{local}},y_{\mathrm{local}},z_{\mathrm{local}})$ attached to the TxUPA.}

\section{ULA Near-Field Distance Calculation}\label{sec:ULA}
In this section, we derive the near-field boundary distance with antennas misalignment \emph{for various ULA configurations}. First, Sec.~\ref{sec:ULA-to-ULA(general)} summarizes the general approach (partially utilized also for more complex UPA-to-UPA scenarios in Sec.~\ref{sec:UPA} below). Then, Sec.~\ref{sec:ULA-to-ULA(on-boresight)} discusses the simplest on-boresight case ($\alpha=0$) with only a single Tx rotation. Later, Sec.~\ref{sec:ULA-to-ULA(off-boresight)} extends the analysis to off-boresight ULA displacements ($\alpha \neq 0$). Finally, Sec.~\ref{sec:ULA-to-point} addresses a specific extreme ULA-to-point setup with off-boresight misalignment.

\subsection{General Approach} \label{sec:ULA-to-ULA(general)}
\rew{To characterize the near-field region for array-to-array links, we adopt a generalized phase-error-based criterion with a maximum allowable phase deviation threshold $\varphi \in (0,\pi]$. The key idea is that the link can be treated as far-field only if \emph{all} received signal components associated with different Tx--Rx element pairs can be \emph{simultaneously phase-aligned} by conventional far-field beamsteering (phase pre-/post-compensation). Accordingly, for a given center-to-center distance $r$, we require that the \emph{maximum residual phase difference} among all received signal components \emph{after accounting for the Tx/Rx phase compensations} does not exceed $\varphi$. We then define the radiative near-field boundary distance $r_{\mathrm F}$ as the smallest $r$ such that this condition holds for all $r \ge r_{\mathrm F}$\footnote{\rew{While this work is motivated by mmWave/THz deployments, the subsequent phase-based characterization and the derived expressions below are not band-specific and, without the loss of generality, may be utilized for other array-based wireless communication systems exploiting other frequency bands (e.g., microwave).}}.}

To achieve this, we accurately derive \emph{the phase delay of each of the received signal components as a combination of three phase adjustments applied to it in such an array-to-array scenario}. The first adjustment is caused by the path the transmitted signal travels between individual Tx and Rx antenna array elements. The second adjustment is applied by the Tx to signal components transmitted by each Tx array element when performing beamsteering\footnote{In this work, beamsteering refers to the phase pre-compensation applied across the antenna elements to steer the main lobe of the transmitted (or received) wavefront toward a desired direction. Note that since our study focuses on determining the near-field boundary distance (starting from which all the far-field solutions, including beamsteering, become fully applicable), we \textbf{may} safely utilize beamsteering-driven phase adjustment in this study.} to direct the transmitted signal toward the Rx in off-boresight setups ($\alpha \neq 0$). The third phase adjustment comes from the Rx performing its own beamsteering to compensate for its own misalignment from the Tx antenna center ($\theta \neq \alpha$).

\subsection{ULA-to-ULA case (on-boresight)} \label{sec:ULA-to-ULA(on-boresight)}
\begin{figure}[!t]
    \centering
    \includegraphics[width=0.8\linewidth]{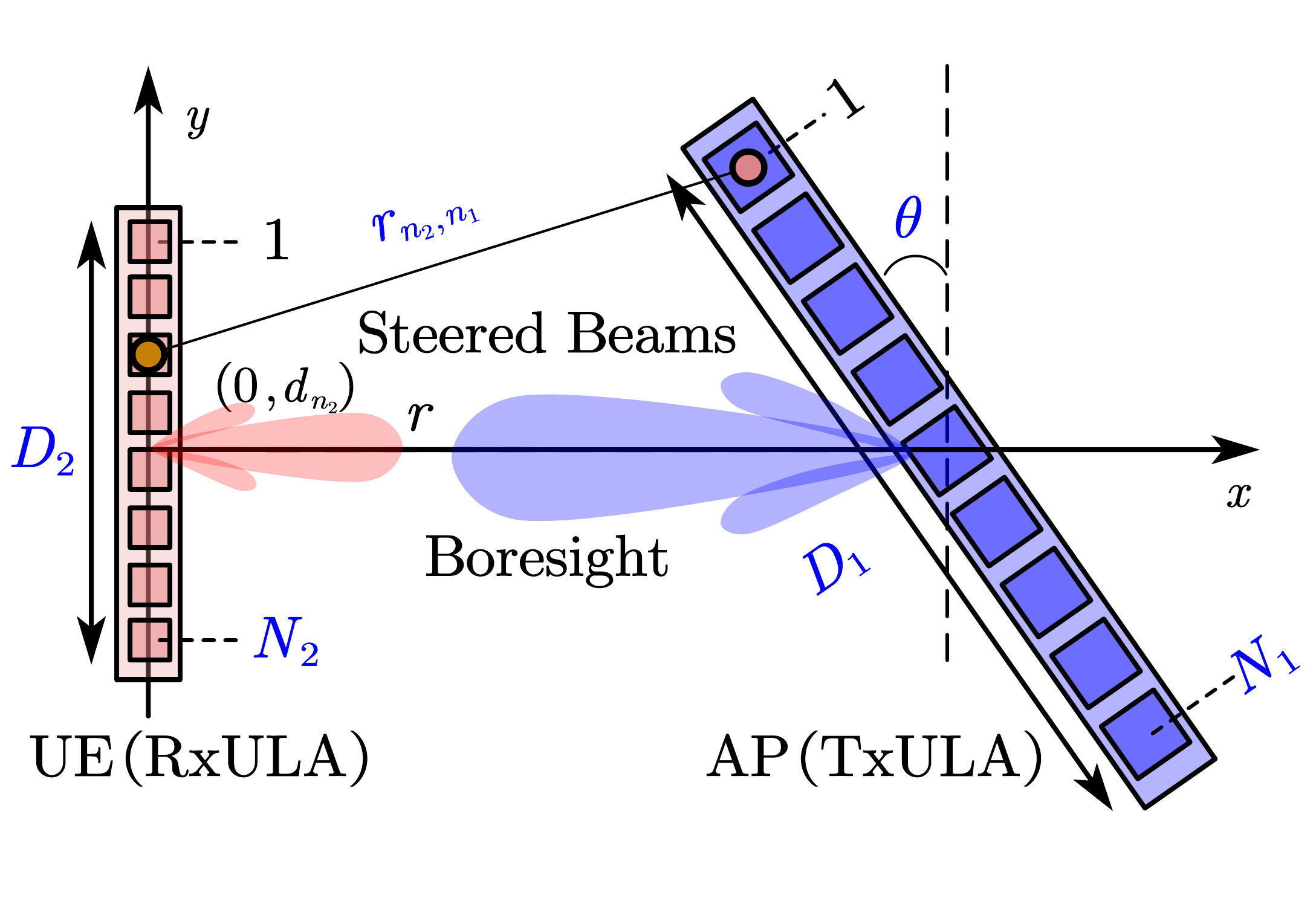}
    \caption{On-boresight ULA scenario ($\alpha = 0$, $\theta \neq 0$, analyzed in Sec.~\ref{sec:ULA-to-ULA(on-boresight)}).}
    \label{fig:ULA-to-ULA_system_model}
\end{figure}

Consider the on-boresight ULA-to-ULA (L2L) scenario in Fig.~\ref{fig:ULA-to-ULA_system_model}, where the TxULA is rotated by an angle $\theta \in [-\pi/2, \pi/2]$. Recall that, while the distance between the centers of the Tx and Rx arrays is set to $r$ (see Fig.~\ref{fig:ULA-to-ULA_system_model}), the exact distances between \emph{individual Tx's and Rx's antenna array elements} are different and depend on the Tx antenna rotation angle $\theta$. Specifically, let $r_{n_2,n_1}$ denote the physical distance between the $n_1$-th Tx element and the $n_2$-th Rx element, where $n_1 = 1, \dots, N_1$ and $n_2 = 1, \dots, N_2$. According to geometric free-space propagation~\cite{sherman1962properties}, the propagation phase between these two array elements, denoted by $\phi_{n_2,n_1}$, is
\begin{equation}\label{equ:system model}
\phi_{n_2,n_1} = - 2\pi r_{n_2,n_1} / \lambda.
\end{equation}

In addition, to steer the beam toward the center of the Rx, the Tx applies a pre-steering phase to the signal transmitted from each Tx antenna element. For $n_1$-th Tx antenna element, this phase delay can be derived as~\cite{balanis2015antenna}
\begin{equation}
\label{equ:compensation}
\tilde{\phi}_{n_1}=- 2\pi d_{n_1}\sin\theta /\lambda,
\end{equation}
where $d_{n_1} \in [-D_1/2, D_1/2]$ is the distance between the $n_1$-th Tx element and the Tx array center. Hence, the total adjusted phase delay between the $n_1$-th element of the Tx array and the $n_2$-th element of the Rx array becomes
\begin{equation}
\label{equ:modified phi}
    \tilde{\phi}_{n_2,n_1}=\phi_{n_2,n_1} + \tilde{\phi}_{n_1}= -{2\pi{r}'_{n_2,n_1}}/{\lambda} ,
\end{equation}
where $r'_{n_2,n_1} \triangleq r_{n_2,n_1} + d_{n_1} \sin\theta$ denotes the so-called \textit{effective distance}, i.e., the equivalent free-space propagation distance after applying the pre-steering phase in~\eqref{equ:compensation}.

We now notice that the near-field criterion (maximum phase discrepancy between the received components not exceeding $\varphi$) corresponds to a maximum variation in effective distance of $\lambda\varphi/{2\pi}$. Therefore, the near-field distance $r = r_{\text{F,on}}^{\text{L2L}}(\theta)$ is defined as the minimum Tx–Rx separation that satisfies
\begin{equation}\label{equ:define near-field}
    r=r_{\text{F,on}}^{\text{L2L}}(\theta),\quad \mathrm{s}.\mathrm{t}. \max_{n_2,n_1} {r}'_{n_2,n_1}-\min_{n_2,n_1} {r}'_{n_2,n_1}= \frac{\lambda\varphi}{2\pi}.
\end{equation}

From the scenario geometry in Fig.~\ref{fig:ULA-to-ULA_system_model}, we obtain
\begin{equation}
r_{{n_2},{n_1}}=\sqrt{\left( r-d_{n_1}\sin \theta \right) ^2+\left( d_{n_1}\cos \theta -d_{n_2} \right) ^2},
\end{equation}
where $d_{n_2} \in [-D_2/2, D_2/2]$ is the distance between the $n_2$-th Rx antenna array element and the center of the Rx array.

By applying the third-order Taylor approximation $\sqrt{1 + x} \approx 1 + \frac{1}{2}x - \frac{1}{8}x^2 + \frac{1}{16}x^{3}$, we obtain the approximation
\begin{equation}\label{equ:U2U_onboresight_taylor}
    r_{n_2,n_1} = r-d_{n_1}\sin \theta +\frac{\left( d_{n_1}\cos \theta -d_{n_2} \right) ^2}{2r}+o\left( \frac{1}{r^2} \right),
\end{equation}
where $ o(\cdot) $ denotes the standard Landau little-o notation. Therefore, the second term (i.e., the first-order or dominant term in $1/r$) becomes the dominant contributor to the residual approximation error. Hence, the effective distance $r'_{n_2,n_1}$ can be well approximated as
\begin{equation}
    r'_{n_2,n_1} \approx r +\frac{\left( d_{n_1}\cos \theta -d_{n_2} \right) ^2}{2r}.
\end{equation}

We then observe from the approximation that the minimum value of $r'_{n_2,n_1}$ is attained when both $d_{n_1} = 0$ and $d_{n_2} = 0$ (i.e., when the Tx and Rx elements are located at their respective array centers), so $\displaystyle\min_{n_2,n_1} {r}'_{n_2,n_1}= r$.

Hence, when enforcing that the effective distance variation remains smaller than $\lambda\varphi/{2\pi}$, as in (\ref{equ:define near-field}), the near-field condition for the on-boresight L2L configuration becomes
\begin{equation}
    \max_{n_2,n_1} \frac{(d_{n_1} \cos\theta - d_{n_2})^2}{2r} \leq \frac{\lambda\varphi}{2\pi}.
\end{equation}

This inequality holds for $d_{n_1} \in [-D_1/2, D_1/2]$ and $d_{n_2} \in [-D_2/2, D_2/2]$. Notably, the maximum value of the left-hand side is achieved when $d_{n_1} = D_1/2$ and $d_{n_2} = -D_2/2$. Substituting these into the inequality, \textbf{the corresponding L2L on-boresight near-field distance is given as}
\begin{equation}
    r_{\text{F,on}}^{\text{L2L}}(\theta)=\frac{\pi(D_1\cos{\theta}+D_2)^2}{4\lambda\varphi}.
\label{equ:U-to-ULA_on_boresight}
\end{equation}

\rew{\textbf{Analysis:} The closed-form boundary in \eqref{equ:U-to-ULA_on_boresight} follows a Fraunhofer-type scaling with an \emph{effective projected aperture} $D_{\mathrm{eff}}(\theta)=D_1\cos\theta + D_2$, i.e., $r_{\mathrm{F,on}}^{\mathrm{L2L}}(\theta)\propto D_{\mathrm{eff}}^2/(\lambda\varphi)$. Hence, the in-plane Tx rotation primarily reduces the projected Tx aperture through $\cos\theta$, which typically \emph{decreases} the near--far field boundary as $|\theta|$ increases, while the dependence is symmetric with respect to $\theta=0$ since $r_{\mathrm{F,on}}^{\mathrm{L2L}}(\theta)=r_{\mathrm{F,on}}^{\mathrm{L2L}}(-\theta)$. When $\varphi=\pi/8$, \eqref{equ:U-to-ULA_on_boresight} reduces to the canonical Fraunhofer criterion~\cite{balanis2015antenna}, i.e., $r_{\mathrm{F,on}}^{\mathrm{L2L}}(\theta)|_{\varphi=\frac{\pi}{8}} = 2(D_1\cos\theta + D_2)^2/\lambda$. Moreover, when $\theta=0$, it further simplifies to $2(D_1 + D_2)^2/\lambda$, which coincides with the no-rotation L2L result in~\cite{petrov2023near}.}

\subsection{ULA-to-ULA case (off-boresight)}\label{sec:ULA-to-ULA(off-boresight)}
We now extend the analysis to the off-boresight scenario illustrated in Fig.~\ref{fig:ULA-to-ULA_system_model(off-boresight)}. The main difference lies in the geometry: specifically, the center of TxULA is no longer aligned with the RxULA boresight, which introduces an additional off-boresight angle $\alpha \in [-\pi/2, \pi/2]$. Due to the inherent symmetry of the array geometry, the TxULA rotation angle $\theta$ can be restricted to the interval $\theta \in \alpha+[-\pi/2, \pi/2]$, which fully captures all distinct configurations without loss of generality. 

\begin{figure}[!t]
    \centering
    \includegraphics[width=0.9\linewidth]{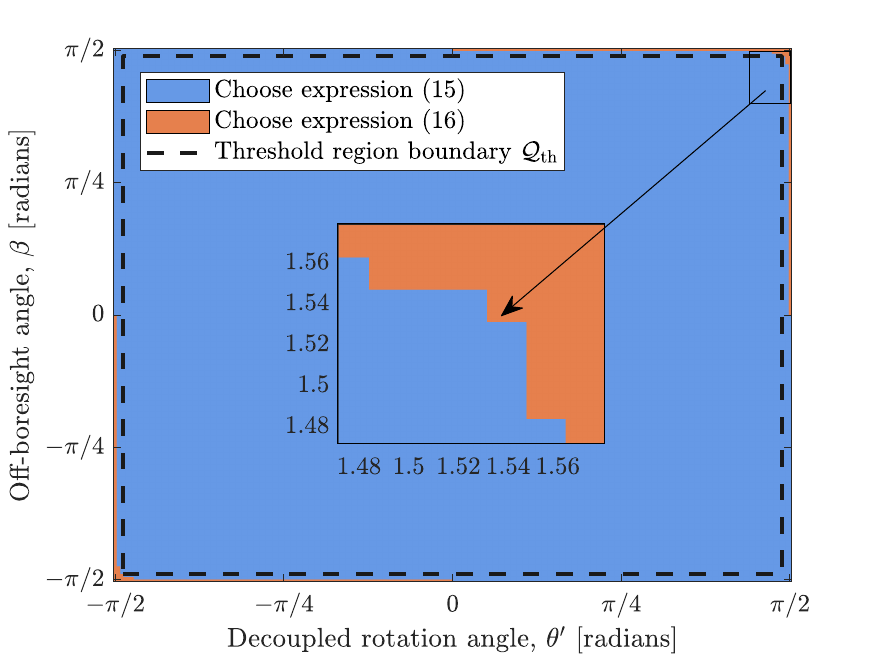}
    \caption{Selection region of expressions~\eqref{equ:U2U_off_a} and~\eqref{equ:U2U_off_b} under different angle pairs $(\theta', \alpha)$, along with the threshold region boundary $\mathcal{Q}_{\mathrm{th}}$ by Proposition~\ref{prop:dominant_a}.}
    \label{fig:U2U_threshold_region}
\end{figure}

To steer the beam toward the RxULA, the TxULA should apply a pre-steering phase $\tilde{\phi}_{n_1}=-{2\pi} d_{n_1}\sin(\theta-\alpha)/{\lambda}$ to antenna elements $n_1$. Compared to the on-boresight case, the off-boresight scenario requires additional pre-steering phase at the Rx array to align the incoming signals. Specifically, each Rx element applies the correction $\tilde{\phi}_{n_2}=-{2\pi} d_{n_2}\sin\alpha/{\lambda}.$

Therefore, the total adjusted phase between the $n_1$-th element of the Tx and the $n_2$-th element of the Rx becomes $\tilde{\phi}_{n_2,n_1}=\phi_{n_2,n_1} + \tilde{\phi}_{n_1} + \tilde{\phi}_{n_2}= -{2\pi} {r}'_{n_2,n_1}/{\lambda}$, where the effective distance ${r}'_{n_2,n_1}$ is given by
\begin{equation}\label{equ:U2U_offboresight_r}
    {r}'_{n_2,n_1} = r_{n_2,n_1} +d_{n_2}\sin\alpha + d_{n_1}\sin(\theta-\alpha).
\end{equation}

According to the geometric relationship, the distance between the $n_1$-th antenna element of the Tx and the $n_2$-th antenna element of the Rx, denoted as $r_{n_2,n_1}$, can be obtained in~\eqref{equ:r_n2.n_1_ULA-to-ULA(off-boresight)} at the top of the next page.
\begin{figure*}[!t]
\begin{equation}\label{equ:r_n2.n_1_ULA-to-ULA(off-boresight)}
       r_{n_2,n_1}=\sqrt{\left( r-d_{n_1}\sin \left( \theta -\alpha \right) -d_{n_2}\sin \alpha \right) ^2+\left( d_{n_1}\cos \left( \theta -\alpha \right) -d_{n_2}\cos \alpha \right) ^2}.
\end{equation}
\end{figure*}

Combining \eqref{equ:U2U_offboresight_r} and \eqref{equ:r_n2.n_1_ULA-to-ULA(off-boresight)}, it follows that $r'_{n_2,n_1} \geq r$ for all $(n_1, n_2)$, which implies that $\min_{n_2,n_1} r'_{n_2,n_1} = r$. Therefore, by denoting $r = r_{\text{F,off}}^{\text{L2L}}(\theta,\alpha)$ as the near-field distance in the off-boresight scenario, and substituting \eqref{equ:r_n2.n_1_ULA-to-ULA(off-boresight)} into the near-field criterion defined in~\eqref{equ:define near-field}, we obtain the expression for $r_{\text{F,off}}^{\text{L2L}}(\theta,\alpha)$ in~\eqref{equ:U2U_off_max_problme}, shown at the top of the next page. 
\begin{figure*}[!t]
\begin{equation}\label{equ:U2U_off_max_problme}
       r_{\text{F,off}}^{\text{L2L}}(\theta,\alpha)=\max_{n_2,n_1} \bigg[\underset{\triangleq f(d_{n_2},d_{n_1})}{\underbrace{ \frac{\pi \left( d_{n_1}\cos \left( \theta -\alpha \right) -d_{n_2}\cos \alpha\right) ^2}{\lambda \varphi}+d_{n_2}\sin \alpha +d_{n_1}\sin\mathrm{(}\theta -\alpha )}} \bigg] -\frac{\lambda \varphi}{4\pi}.
\end{equation}
\hrulefill
\end{figure*}

The expression inside the maximization in~\eqref{equ:U2U_off_max_problme} can be viewed as a convex function $f(d_{n_2}, d_{n_1})$ over the compact convex set $[{-D_1}/{2}, {D_1}/{2}] \times [{-D_2}/{2}, {D_2}/{2}]$. By the extreme value theorem for convex functions, the maximum is attained at one of the vertices of this rectangular region. Therefore, ignoring the wavelength-related term ${\lambda \varphi}/{4\pi}$ in~\eqref{equ:U2U_off_max_problme}, which is upper-bounded by ${\lambda}/{4}$ due to $\varphi \in (0,\pi]$, \textbf{the corresponding L2L off-boresight near-field distance is given as}

\begin{equation}\label{equ:L2L_exact}
    r_{\text{F,off}}^{\text{L2L}}(\theta,\alpha) = \max \left( r_{(\text{a})}^{\text{L2L}}(\theta,\alpha),\, r_{(\text{b})}^{\text{L2L}} (\theta,\alpha)\right),
\end{equation}
where
\begin{equation}\label{equ:U2U_off_a}
    \begin{split}
        r_{(\text{a})}^{\text{L2L}}(\theta,\alpha)&=\frac{\pi\left( D_1\cos (\theta-\alpha)+D_2\cos \alpha \right) ^2}{4\lambda\varphi}\\
        &+\left|{D_1\sin (\theta-\alpha) -D_2\sin \alpha }\right|/2
    \end{split}
\end{equation}
and
\begin{equation}\label{equ:U2U_off_b}
   \begin{split}
       r_{(\text{b})}^{\text{L2L}}(\theta,\alpha)&=\frac{\pi \left( D_1\cos (\theta-\alpha) -D_2\cos \alpha \right) ^2}{4\lambda \varphi}\\
       &+\left|{D_1\sin (\theta-\alpha) +D_2\sin \alpha}\right|/2.
   \end{split}
\end{equation}

\begin{proposition}
\label{prop:dominant_a}
The off-boresight near-field distance is dominated by case~(a), i.e., $r_{\mathrm{F,off}}^{\mathrm{L2L}}(\theta,\alpha)=r_{(\mathrm{a})}^{\mathrm{L2L}}(\theta,\alpha)$, provided that the angles $\theta'$ and $\alpha$ satisfy $|\theta'|,\, |\alpha| < \vartheta_{\mathrm{th}},$ where $\theta'\triangleq\theta-\alpha$ and the threshold angle $\vartheta_{\mathrm{th}}$ is defined as
\begin{equation}
    \vartheta_{\mathrm{th}} = \sin^{-1}\left( \frac{-1+\sqrt{1+4\kappa^2}}{2\kappa} \right)
\end{equation}
with $\kappa = {\pi\max \left( D_1, D_2 \right) }/{(\lambda \varphi)}$.
\end{proposition}

\begin{IEEEproof}
The proof is provided in Appendix~\ref{appendix:proof_dominant_a}.
\end{IEEEproof}

It is worth noting that the threshold angle $\vartheta_{\mathrm{th}}$ in Proposition~\ref{prop:dominant_a} depends on the parameter $\kappa = {\pi \max(D_1, D_2)}/{(\lambda \varphi)}$. In practical settings, this quantity is typically large, since the array aperture $\max(D_1, D_2)$ is much greater than the product of the wavelength $\lambda$ and the phase deviation parameter $\varphi$. As a result, the value of $\vartheta_{\mathrm{th}}$ approaches $\pi/2$, which implies that the theoretical threshold region $\mathcal{Q}_{\mathrm{th}} \triangleq (-\vartheta_{\mathrm{th}}, \vartheta_{\mathrm{th}})^2$ spans almost the entire angular domain.

To validate the above theoretical insight, the practical relevance of Proposition~\ref{prop:dominant_a} is illustrated in  Fig.~\ref{fig:U2U_threshold_region}, where the dominating regions of expressions~\eqref{equ:U2U_off_a} and~\eqref{equ:U2U_off_b} are illustrated for varying angle pairs $(\theta', \alpha)$. The dashed boundary corresponds to the theoretical threshold region $\mathcal{Q}_{\mathrm{th}} $. The simulation is conducted with parameters $D_1 = 0.1\,\mathrm{m}$, $D_2 = 0.05\,\mathrm{m}$, $\varphi = \pi/8$, and carrier frequency $f = 300\,\mathrm{GHz}$, yielding $\lambda = c/f \approx 10^{-3}\,\mathrm{m}$ and $\kappa = {\pi \cdot 0.1}/({10^{-3} \cdot \pi/8}) = 800$. Substituting these values into the expression of $\vartheta_{\mathrm{th}}$ in Proposition~\ref{prop:dominant_a} gives $\vartheta_{\mathrm{th}} \approx \sin^{-1}(0.998) \approx 85.9^\circ$, and thus $\mathcal{Q}_{\mathrm{th}} \approx (-85.9^\circ, 85.9^\circ)^2$. This indicates that case~(a) dominates within nearly the entire angular domain, as confirmed in the figure, where only a narrow corner region lies outside $\mathcal{Q}_{\mathrm{th}}$ and requires the use of expression~\eqref{equ:U2U_off_b}.

While Fig.~\ref{fig:U2U_threshold_region} confirms the angular dominance of case~(a) from a theoretical perspective, a more detailed numerical result is presented in Fig.~\ref{fig:chooseExp} (see Sec.~\ref{sec:results}). This result verifies that expression~\eqref{equ:U2U_off_a} alone suffices to accurately characterize the near-field distance across the entire angular domain of interest. Furthermore, we note that the second term in \eqref{equ:U2U_off_a}, which accounts for the aperture-induced phase variation, is upper bounded by $(D_1 + D_2)/2$ and thus remains at the aperture-size level. Since this contribution is generally negligible in practical scenarios, \textbf{the near-field distance can be further approximated as}
\begin{equation}
    \tilde{r}_{\text{F,off}}^{\text{L2L}}(\theta,\alpha)\approx\frac{\pi(D_1\cos(\theta-\alpha)+D_2\cos{\alpha})^2}{4\lambda\varphi}.
\label{equ:ULA-to-ULA_off_boresight}
\end{equation}

\rew{\textbf{Analysis:} The approximation \eqref{equ:ULA-to-ULA_off_boresight} reveals that the dominant off-boresight L2L boundary is governed by the \emph{effective projected aperture} $D_1\cos(\theta-\alpha)+D_2\cos\alpha$. This makes the role of misalignment transparent: the boundary depends on the \emph{relative} Tx rotation mismatch $\theta'=\theta-\alpha$ (rather than $\theta$ alone), so aligning the Tx array orientation with the link direction (i.e., $\theta\approx\alpha$) maximizes the projected Tx aperture and can significantly change near-field distance compared to a no-rotation baseline. When $\varphi=\pi/8$, \eqref{equ:ULA-to-ULA_off_boresight} reduces to the canonical Fraunhofer criterion~\cite{balanis2015antenna}, i.e., $\tilde{r}_{\text{F,off}}^{\text{L2L}}(\theta,\alpha)|_{\varphi=\frac{\pi}{8}} = 2(D_1\cos(\theta-\alpha) + D_2\cos \alpha)^2/\lambda$.}

\subsection{ULA-to-point case (off-boresight)}\label{sec:ULA-to-point}
\begin{figure}[!t]
    \centering
    \includegraphics[width=0.6\linewidth]{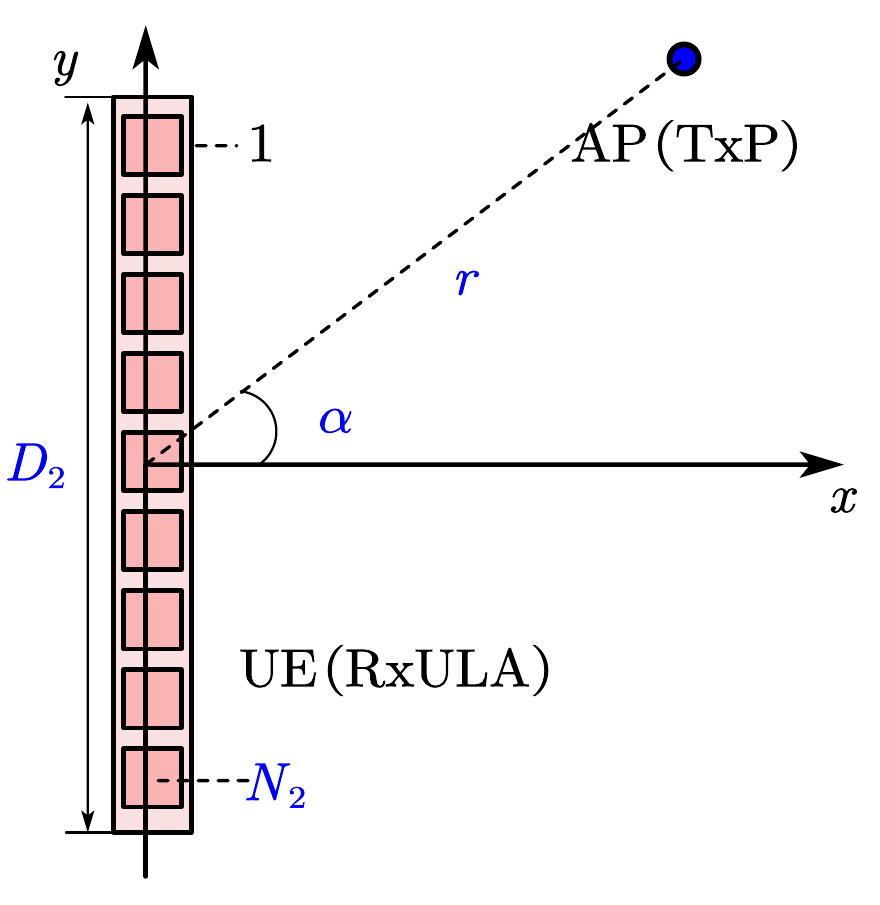}
    \caption{Off-boresight ULA-to-point system model (Sec.~\ref{sec:ULA-to-point}).}
    \label{fig:ULA-to-point_system_model}
\end{figure}

We next consider a special case of the off-boresight configuration, namely the ULA-to-point scenario, as in Fig.~\ref{fig:ULA-to-point_system_model}. In this case, the TxULA is reduced to a single point source (TxP, $N_{1} = 1$). The distance between the $n_2$-th antenna element of the RxULA and the TxP is given by
\begin{equation}\label{equ:U2P_off_r_n}
    r_{n_2} = \sqrt{\left( r-d_{n_2}\sin \alpha \right) ^2+d_{n_2}^{2}\cos ^2\alpha} \text{,}
\end{equation}
where $d_{n_2} \in [-D_2/2, D_2/2]$ denotes the position of the $n_2$-th antenna element relative to the center of the array. 

To coherently combine the incoming wavefront from the point source, each receive antenna element applies a pre-steering phase given by
\begin{equation}
    \tilde{\phi}_{n_2} = -{2\pi} d_{n_2} \sin\alpha/{\lambda}.
\end{equation}

Therefore, the total phase delay can be written as
\begin{equation}
    {\phi}'_{n_2} = \phi_{n_2} + \tilde{\phi}_{n_2} = -{2\pi} r'_{n_2}/{\lambda},
\end{equation}
where $r'_{n_2} = r_{n_2} + d_{n_2} \sin\alpha$ denotes the \textit{effective distance} between the ${n_2}$-th antenna element and the point source.

\rew{Following the near-field criterion in~\eqref{equ:define near-field}, we define the off-boresight near-field distance $r_{\mathrm{F,off}}^{\mathrm{L2O}}(\alpha)$ as}
\begin{equation}\label{equ:def_U2P_off}
\rew{r=r_{\mathrm{F,off}}^{\mathrm{L2O}}(\alpha),
\quad \text{s.t.}\quad
\max_{{n_2}} r'_{n_2}
-\min_{{n_2}} r'_{n_2}
\le \frac{\lambda\varphi}{2\pi}.}
\end{equation}

Combining this expression with the geometric distance in~\eqref{equ:U2P_off_r_n}, it follows that $r'_{n_2} \geq r$ for all $d_{n_2} \in [-D_2/2, D_2/2]$, with equality attained at $d_{n_2} = 0$. Hence, $\min_{n_2} r'_{n_2} = r$. Substituting this into the near-field criterion defined in~\eqref{equ:define near-field}, and denoting $r = r_{\text{F,off}}^{\mathrm{L2O}}(\alpha)$, a straightforward algebraic manipulation yields
\begin{equation}\label{equ:U2P_r_beta_init}
    r_{\text{F,off}}^{\mathrm{L2O}}(\alpha)= \max_{{n_2}} \left( \frac{\pi d_{n_2}^2 \cos^2\alpha}{\lambda \varphi} - d_{n_2} \sin\alpha \right).
\end{equation}

Note that the maximum in~\eqref{equ:U2P_r_beta_init} is attained at $d_{n_2} = -\operatorname{sgn}(\sin\alpha){D}_2/{2}$, where $\operatorname{sgn}(\cdot)$ denotes the sign function. Therefore, \textbf{the near-field distance for the off-boresight ULA-to-point configuration is given by}
\begin{equation}\label{equ:L2O_off_closed}
    r_{\text{F,off}}^{\mathrm{L2O}}(\alpha) = \frac{\pi D_2^2 \cos^2\alpha}{4\lambda \varphi} + \frac{D_2}{2} |\sin\alpha| .
\end{equation}

We again note that the second term is upper bounded by ${D_2}/{2}$, corresponding to the aperture size of the Rx array. As this contribution is generally negligible in practical scenarios, \textbf{ the near-field distance thus becomes}
\begin{equation}
    \tilde{r}_{\text{F,off}}^{\mathrm{L2O}}(\alpha) \approx \frac{\pi D_2^2 \cos^2\alpha}{4\lambda \varphi}.
\end{equation}

\rew{\textbf{Analysis:} The ULA-to-point near-field boundary admits a simple interpretation: the dominant term in $\tilde{r}_{\text{F,off}}^{\mathrm{L2O}}(\alpha)$ scales with the \emph{projected} Rx aperture $D_2\cos\alpha$, i.e., $\tilde{r}_{\text{F,off}}^{\mathrm{L2O}}(\alpha)\propto (D_2\cos\alpha)^2/(\lambda\varphi)$. Hence, increasing the off-boresight angle magnitude $|\alpha|$ reduces the projected aperture and typically \emph{decreases} the near-field boundary, which is consistent with the intuition that a more grazing incidence produces weaker aperture-induced phase curvature along the Rx. When $\alpha=0$, the boundary reduces to $r_{\text{F,off}}^{\mathrm{L2O}}(0)=\pi D_2^2/(4\lambda\varphi)$, consistent with~\cite{selvan2017fraunhofer}.}

\section{\rew{UPA Near-Field Distance Calculation}}\label{sec:UPA}

Exploiting the key results for the ULA case (Sec.~\ref{sec:ULA} above) and the general approach in Sec.~\ref{sec:ULA-to-ULA(general)}, we now extend the analysis to more complex UPA configurations. Sec.~\ref{sec:UPA-to-UPA(on-boresight)} considers the on-boresight case ($\alpha = 0, \phi \neq 0, \theta \neq 0$). Sec.~\ref{sec:UPA-to-UPA(off-boresight,1)} discusses the off-boresight setup with a single-angle Tx rotation ($\alpha \neq 0, \phi = 0, \theta \neq 0$), while Sec.~\ref{sec:UPA-to-UPA(off-boresight,2)} generalizes the analysis to dual-angle Tx rotations ($\alpha \neq 0, \phi \neq 0, \theta \neq 0$). \rew{To illustrate the extensibility of the proposed geometry-driven framework beyond same-type arrays, Sec.~\ref{sec:mixed_P2L_on} provides representative mixed UPA--ULA (P2L) on-boresight examples.} Finally, Sec.~\ref{sec:UPA-to-point} addresses the UPA-to-point setup with off-boresight displacement ($\alpha \neq 0, N_{1} = 1$).

\subsection{UPA-to-UPA case (on-boresight)}\label{sec:UPA-to-UPA(on-boresight)}
\begin{figure}[!t]
    \centering
    \includegraphics[width=0.9\linewidth]{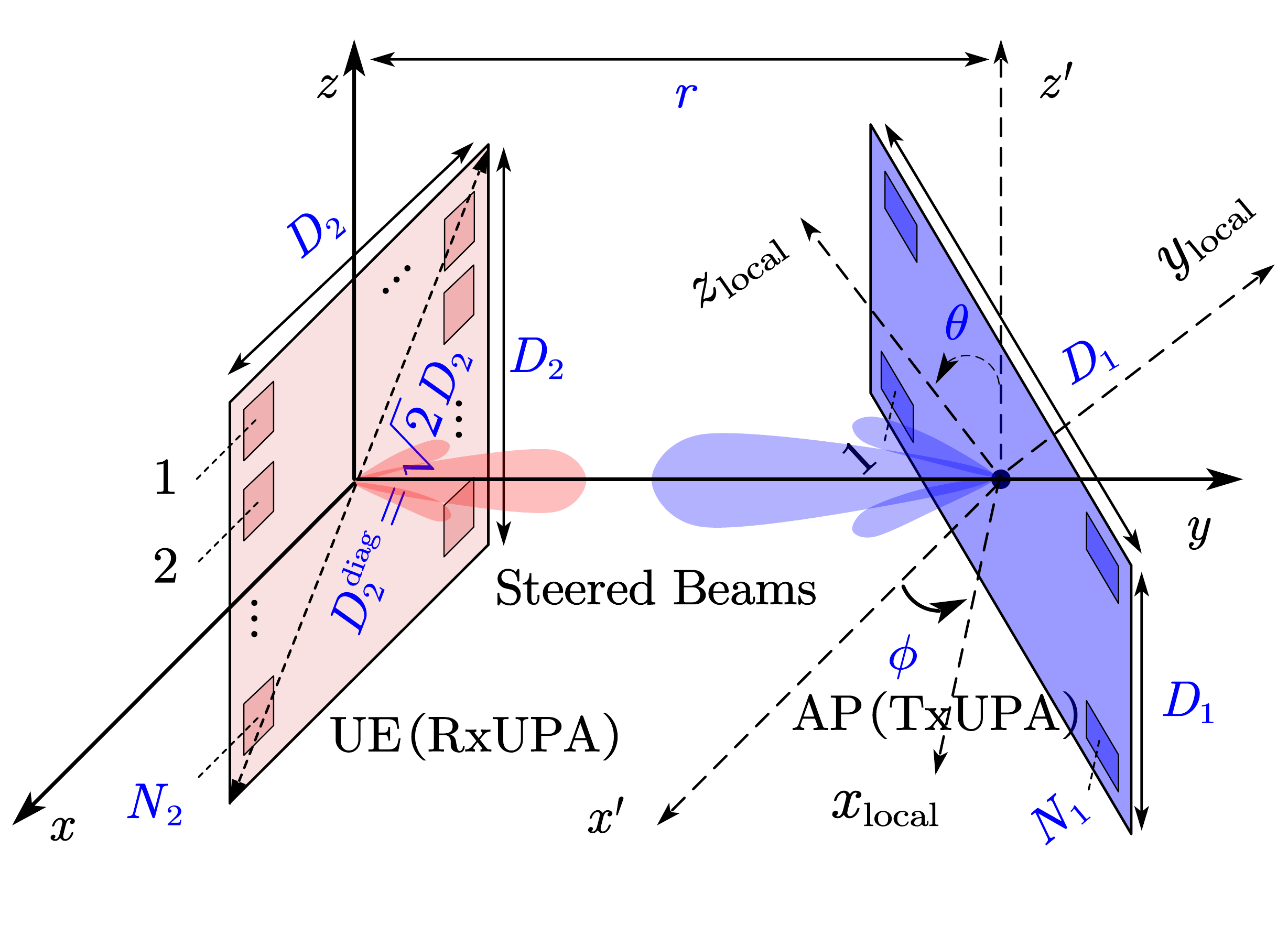}
    \caption{On-boresight UPA-to-UPA setup with TxUPA rotation (Sec.~\ref{sec:UPA-to-UPA(on-boresight)}).}
    \label{fig:UPA-to-UPA_system_model(on-boresight)}
\end{figure}

We now extend the previous analysis to the on-boresight UPA-to-UPA (P2P) configuration, as illustrated in Fig.~\ref{fig:UPA-to-UPA_system_model(on-boresight)}. Similar to the ULA case, the RxUPA is fixed in the $xz$-plane, and we denote the position of its $(m,n)$-th element as $\mathbf{d}_{m,n} = (d_m, 0, d_n)^{\top}$, where $d_m, d_n \in [-D_2/2, D_2/2]$. As described in Sec.~\ref{sec:system_model}, to model realistic misalignments due to small-scale motion, the TxUPA undergoes a composite rotation defined by the matrix $\mathbf{R}(\theta, \phi) = \mathbf{R}_z(\phi)\mathbf{R}_x(\theta)$, where the component rotation matrices are

\begin{align}
\mathbf{R}_x(\theta) &=
\begin{bmatrix}
1 & 0 & 0 \\
0 & \cos\theta & -\sin\theta \\
0 & \sin\theta & \cos\theta
\end{bmatrix}, \label{rotation_x} \\
\mathbf{R}_z(\phi) &=
\begin{bmatrix}
\cos\phi & -\sin\phi & 0 \\
\sin\phi & \cos\phi & 0 \\
0 & 0 & 1
\end{bmatrix}. \label{rotation_z}
\end{align}

Let the local position of the $(i,j)$-th TxUPA element be $\mathbf{d}_{i,j} = (d_i, 0, d_j)^\top$, with $d_i, d_j \in [-D_1/2, D_1/2]$. After applying the composite rotation $\mathbf{R}(\theta,\phi) = \mathbf{R}_z(\phi)\mathbf{R}_x(\theta)$, its global position becomes
\begin{equation}
\mathbf{d}'_{i,j} = \mathbf{R}(\theta,\phi)\mathbf{d}_{i,j} + r\mathbf{e}_y,
\end{equation}
where $\mathbf{e}_y = (0,1,0)^{\top}$ is the unit vector along the $y$-axis and $r$ denotes the center-to-center Tx–Rx separation.

The distance between the $(m,n)$-th Rx element and the $(i,j)$-th Tx element is
\begin{equation}
r^{i,j}_{m,n} = \left\| \mathbf{d}_{m,n} - \mathbf{d}'_{i,j} \right\|.
\end{equation}

To steer the beam from TxUPA to RxUPA along the $-y$ axis, pre-steering phase is applied at the Tx, given by
\begin{equation}
\tilde{\phi}_{i,j} = -\frac{2\pi}{\lambda} \tilde{\mathbf{d}}_{i,j}^{\top} \cdot (-\mathbf{e}_y) = \frac{2\pi}{\lambda} \tilde{\mathbf{d}}_{i,j}^{\top} \cdot \mathbf{e}_y,
\end{equation}
where $\tilde{\mathbf{d}}_{i,j} = \mathbf{R}(\theta,\phi)\mathbf{d}_{i,j}$.

The effective propagation distance after pre-steering is
\begin{equation}
r'^{i,j}_{m,n} = r^{i,j}_{m,n} + \tilde{\mathbf{d}}^{\top}_{i,j} \cdot \mathbf{e}_y.
\end{equation}

We define the near-field region using the same phase-error-based criterion as in Section~\ref{sec:ULA-to-ULA(on-boresight)}, requiring the maximal variation in $r'^{i,j}_{m,n}$ to be within $\lambda\varphi/2\pi$, i.e., 
\begin{equation}
r = r_{\theta,\phi,0}^{\text{P2P}}, \quad \mathrm{s}.\mathrm{t}. \max_{m,n,i,j} {r}'^{i,j}_{m,n} - \min_{m,n,i,j} {r}'^{i,j}_{m,n} = \frac{\lambda\varphi}{2\pi}.
\end{equation}

By geometry, and denoting $\mathbf{v}^{i,j}_{m,n} = \mathbf{d}_{m,n} - \tilde{\mathbf{d}}_{i,j}$, we expand
\begin{equation}
r^{i,j}_{m,n} = \sqrt{r^2 - 2r\, \mathbf{e}_y^\top \mathbf{v}^{i,j}_{m,n} + \|\mathbf{v}^{i,j}_{m,n}\|^2}.
\end{equation}

Applying the third-order Taylor approximation yields
\begin{equation}
r^{i,j}_{m,n} \approx r - \mathbf{e}_y^\top \mathbf{v}^{i,j}_{m,n} + \frac{ \|\mathbf{v}^{i,j}_{m,n}\|^2 - (\mathbf{e}_y^\top \mathbf{v}^{i,j}_{m,n})^2 }{2r}.
\end{equation}
Then,
\begin{equation}\label{equ:P2P_inequality}
r'^{i,j}_{m,n} \approx r + \frac{ \left( v_{m,n}^{(x)i,j} \right)^2 + \left( v_{m,n}^{(z)i,j} \right)^2 }{2r},
\end{equation}
where
\begin{align}
v_{m,n}^{(x)i,j} &= d_m - d_i \cos\phi - d_j \sin\phi \sin\theta, \\
v_{m,n}^{(z)i,j} &= d_n - d_j \cos\theta.
\end{align}

Noting that the minimum value is achieved when $d_i=d_j=d_m=d_n=0$, the near-field condition becomes
\begin{equation}
\max_{m,n,i,j} \frac{ \left( v_{m,n}^{(x)i,j} \right)^2 + \left( v_{m,n}^{(z)i,j} \right)^2 }{2r} \leq \frac{\lambda\varphi}{2\pi}.
\end{equation}

The worst-case configuration occurs when $d_i = {D_1}/{2}$, $d_j = -\operatorname{sgn}(\sin\phi \sin\theta){D_1}/{2}$, $d_m = {D_2}/{2}$, $ d_n = -\operatorname{sgn}(d_j) {D_2}/{2}$. Therefore, \textbf{the corresponding P2P on-boresight near-field distance is then given by}
\begin{equation}
\begin{split}
r_{\text{F,on}}^{\text{P2P}}(\theta,\phi)&=\frac{\pi \left( D_1\cos \theta +D_2 \right) ^2}{4\lambda \varphi}
\\
&+\frac{\pi \left( D_1\left( \cos \phi +|\sin \theta \sin \phi | \right) +D_2 \right) ^2}{4\lambda \varphi}.
\end{split}
\label{equ:UPA-to-UPA_onboresight}
\end{equation}

\rew{\textbf{Analysis:} The on-boresight P2P boundary in \eqref{equ:UPA-to-UPA_onboresight} is the \emph{sum of two Fraunhofer-type contributions}, which reflects the two transverse dimensions of a UPA aperture. The first term depends on $D_1\cos\theta + D_2$ and captures the vertical rotation impact through a simple projection (similar to the ULA behavior). In contrast, the second term depends on $D_1(\cos\phi + |\sin\theta\sin\phi|)+D_2$ and explicitly couples azimuth and elevation, which implies that misalignment can alter the effective aperture in a \emph{non-separable} manner. This angle coupling is absent in the ULA case and is a key reason why ignoring misalignment can be inaccurate for UPAs under general 3D orientations.} When $\varphi=\pi/8$, \eqref{equ:UPA-to-UPA_onboresight} reduces to a 2D generalization of the classical Fraunhofer distance, and when $\theta = \phi = 0$, it further simplifies to
\begin{equation}\label{equ:traditional_upa}
    r_{\text{F,on}}^{\text{P2P}}(0,0)|_{\varphi=\frac{\pi}{8}} = \frac{4(D_1 + D_2)^2}{\lambda},
\end{equation}
consistent with the previous P2P near-field distance in \cite{petrov2023near}.

\subsection{UPA-to-UPA case (off-boresight) with single-angle rotation}\label{sec:UPA-to-UPA(off-boresight,1)}

We now extend the analysis to the off-boresight P2P scenario. As a first step, we consider the single-angle rotation case of Fig.~\ref{fig:UPA-to-UPA_system_model(off-boresight,two_angles)} with the TxUPA rotation angle $\phi=0$. In this configuration, the RxUPA remains fixed in the $xz$-plane, while the TxUPA is rotated around the $x$-axis by an angle $\theta \in \alpha+[-\pi/2 , \pi/2 ]$ and steers its beam toward an off-boresight angle $\alpha \in [-\pi/2, \pi/2]$. Let $\mathbf{e}_r = (0, \cos\alpha, \sin\alpha)^\top$ denote the unit vector along the line connecting the array centers. The global position of the $(i,j)$-th element on the TxUPA is then given by
\begin{equation}
\mathbf{d}_{i,j} = r \mathbf{e}_r + \mathbf{R}_x(\theta) \mathbf{d}^{xz}_{i,j},
\end{equation}
where $\mathbf{R}_x(\theta)$ is the $x$-axis rotation matrix defined in \eqref{rotation_x} and $\mathbf{d}^{xz}_{i,j} = (d_i, 0, d_j)^\top$ is the local coordinate of the Tx element.

Accordingly, the distance between the $(m,n)$-th Rx element and the $(i,j)$-th Tx element is
\begin{equation}\label{equ:UPA-to-UPA_r_define_off-boresight}
r^{i,j}_{m,n} = \left\| \mathbf{d}_{m,n} - \mathbf{d}_{i,j} \right\|,
\end{equation}
which can be rewritten as
\begin{equation}
r^{i,j}_{m,n} = \sqrt{r^2 - 2r \mathbf{e}_r^\top \mathbf{u}^{i,j}_{m,n} + \left\| \mathbf{u}^{i,j}_{m,n} \right\|^2},
\end{equation}
where $\mathbf{u}^{i,j}_{m,n} = \mathbf{d}_{m,n} - \mathbf{R}_x(\theta) \mathbf{d}^{xz}_{i,j}$ and $\mathbf{d}_{m,n} = (d_m, 0, d_n)^\top$.

To steer the beam toward the Rx, the Tx applies a pre-steering phase aligned with $-\mathbf{e}_r$, given by
\begin{equation}
\tilde{\phi}_{i,j} = - \frac{2\pi }{\lambda}\left( \mathbf{R}_x(\theta) \mathbf{d}^{xz}_{i,j} \right)^{\top} \cdot (-\mathbf{e}_r).
\end{equation}

Similar to the off-boresight ULA case, the off-boresight UPA configuration also requires an additional pre-steering phase at the Rx array, given by
\begin{equation}\label{equ:RxUPA_compensation}
\tilde{\phi}_{m,n} = - \frac{2\pi }{\lambda}\mathbf{d}^{\top}_{m,n} \cdot \mathbf{e}_r.
\end{equation}

As a result, the total adjusted phase difference corresponds to the effective distance
\begin{equation}
{r}'^{i,j}_{m,n} = r^{i,j}_{m,n} + \mathbf{d}^{\top}_{m,n} \cdot \mathbf{e}_r - \left( \mathbf{R}_x(\theta) \mathbf{d}^{xz}_{i,j} \right)^{\top} \cdot \mathbf{e}_r.
\end{equation}

Applying the third-order Taylor approximation $\sqrt{1 + x} \approx 1 + \frac{1}{2}x - \frac{1}{8}x^2 + \frac{1}{16}x^3$, we obtain
\begin{equation}
r^{i,j}_{m,n} = r - \mathbf{e}_r^\top \mathbf{u}^{i,j}_{m,n} + \frac{\|\mathbf{u}^{i,j}_{m,n}\|^2 - \left( \mathbf{e}_r^\top \mathbf{u}^{i,j}_{m,n} \right)^2}{2r} + o\left( \frac{1}{r^2} \right).
\end{equation}

By defining $\mathbf{u}^{i,j}_{m,n} = ( u^{(x)i,j}_{m,n}, u^{(y)i,j}_{m,n}, u^{(z)i,j}_{m,n})^\top$, we have
\begin{equation}
    \mathbf{e}_r^\top \mathbf{u}^{i,j}_{m,n} = \cos\alpha \, u^{(y)i,j}_{m,n}+ \sin\alpha \, u^{(z)i,j}_{m,n},
\end{equation}
\begin{equation}
    \begin{split}
        &\left\| \mathbf{u}^{i,j}_{m,n} \right\|^2 - \left( \mathbf{e}_r^\top \mathbf{u}^{i,j}_{m,n} \right)^2 = ( u^{(x)i,j}_{m,n} )^2 \\
        &+ ( \sin\alpha \, u^{(y)i,j}_{m,n} - \cos\alpha \, u^{(z)i,j}_{m,n} )^2,
    \end{split}
\end{equation}
where $u^{(x)i,j}_{m,n} = d_m - d_i$, $u^{(y)i,j}_{m,n} = d_j \sin\theta$, and $u^{(z)i,j}_{m,n} = d_n - d_j \cos\theta$.

To quantify the near-field distance, we apply the same phase-error-based criterion as in previous sections. Noting that the minimum effective distance is again attained at the array centers, we require
\begin{equation}\label{equ:define_near-field_U2U_offboresight}
\max_{m,n,i,j}\frac{\left( u^{(x)i,j}_{m,n} \right)^2+\left( \sin\alpha \, u^{(y)i,j}_{m,n} - \cos\alpha \, u^{(z)i,j}_{m,n} \right)^2}{2r} \leq \frac{\lambda\varphi}{2\pi}.
\end{equation}

By applying trigonometric identities, we observe that
\begin{equation}
\sin\alpha\, u_{m,n}^{(y)(i,j)} - \cos\alpha\, u_{m,n}^{(z)(i,j)} = d_j \cos\theta'-d_n \cos\alpha,
\end{equation}
which implies that the maximum deviation is attained at $d_m = D_2/2$, $d_i = -D_1/2$, $d_n = -D_2/2$, and $d_j = D_1/2$. Therefore, \textbf{the corresponding P2P off-boresight near-field distance (with single-angle rotation) is given by}
\begin{equation}
\label{equ:UPA-to-UPA_off-boresight_0}
    r_{\text{F,off}}^{\text{P2P}}(\theta,\alpha) = \frac{\pi\left( D_2 + D_1 \right)^2}{4\lambda\varphi}
    + \frac{\pi\left(  D_1\cos (\theta-\alpha) + D_2\cos \alpha \right)^2}{4\lambda\varphi}.
\end{equation}

\rew{\textbf{Analysis:} The single-angle off-boresight P2P boundary in \eqref{equ:UPA-to-UPA_off-boresight_0} decomposes into two terms with clear physical meanings. The angle-independent term $\pi(D_2+D_1)^2/(4\lambda\varphi)$ stems from the intrinsic second aperture dimension of the UPA that remains unaffected by the considered vertical rotation, while the second term $\pi( D_1\cos(\theta-\alpha) + D_2\cos\alpha )^2/(4\lambda\varphi)$ mirrors the ULA-like projected-aperture behavior along the other transverse dimension. Therefore, compared to the ULA counterpart \eqref{equ:ULA-to-ULA_off_boresight}, UPAs generally exhibit larger near-field distance and a different sensitivity pattern to rotations, since one dimension can contribute a non-negligible baseline even when the other dimension is strongly misaligned.}

\subsection{UPA-to-UPA case (off-boresight) with dual-angle rotation}\label{sec:UPA-to-UPA(off-boresight,2)}
We now consider the more general P2P off-boresight scenario shown in Fig.~\ref{fig:UPA-to-UPA_system_model(off-boresight,two_angles)}, where the TxUPA experiences a composite rotation defined by the matrix $\mathbf{R}(\theta, \phi) = \mathbf{R}_z(\phi)\mathbf{R}_x(\theta)$, resulting in a dual-angle misalignment. The RxUPA remains fixed in the $xz$-plane, and the TxUPA steers its beam toward an off-boresight angle $\alpha$. Therefore, the global position of the $(i,j)$-th Tx element is
\begin{equation}
\mathbf{d}_{i,j} = r \mathbf{e}_r + \mathbf{R}(\theta,\phi) \mathbf{d}^{xz}_{i,j}.
\end{equation}

Accordingly, the distance between the $(m,n)$-th Rx element and the $(i,j)$-th Tx element is
\begin{equation}
r^{i,j}_{m,n} = \left\| \mathbf{d}_{m,n} - \mathbf{d}_{i,j} \right\| = \sqrt{r^2 - 2r \mathbf{e}_r^\top \mathbf{u}^{i,j}_{m,n} + \left\| \mathbf{u}^{i,j}_{m,n} \right\|^2},
\end{equation}
where $\mathbf{u}^{i,j}_{m,n} = \mathbf{d}_{m,n} - \mathbf{R}(\theta,\phi) \mathbf{d}^{xz}_{i,j}$ and $\mathbf{d}_{m,n} = (d_m, 0, d_n)^\top$.

To ensure the transmitted signal aligns with $-\mathbf{e}_r$, the Tx applies a pre-steering phase
\begin{equation}
\tilde{\phi}_{i,j} = -\frac{2\pi}{\lambda} \left( \mathbf{R}(\theta,\phi) \mathbf{d}^{xz}_{i,j} \right)^\top \cdot (-\mathbf{e}_r).
\end{equation}

The RxUPA applies a similar pre-steering phase as defined in~\eqref{equ:RxUPA_compensation}. The total adjusted propagation distance can thus be written as
\begin{equation}
r'^{i,j}_{m,n} = r^{i,j}_{m,n} + \mathbf{d}_{m,n}^\top \cdot \mathbf{e}_r - \left( \mathbf{R}(\theta,\phi)\mathbf{d}^{xz}_{i,j} \right)^\top \cdot \mathbf{e}_r.
\end{equation}

We redefine $\mathbf{u}^{i,j}_{m,n} = ( u^{(x)i,j}_{m,n}, u^{(y)i,j}_{m,n}, u^{(z)i,j}_{m,n})^\top$, where
\begin{align}
u^{(x)i,j}_{m,n} &= d_m - d_i \cos\phi - d_j \sin\phi \sin\theta,\label{equ:A2A_ux} \\
u^{(y)i,j}_{m,n} &= -d_i \sin\phi + d_j \cos\phi \sin\theta,\label{equ:A2A_uy} \\
u^{(z)i,j}_{m,n} &= d_n - d_j \cos\theta.\label{equ:A2A_uz}
\end{align}

Substituting~\eqref{equ:A2A_ux}–\eqref{equ:A2A_uz} into the near-field criterion~\eqref{equ:define_near-field_U2U_offboresight}, we define the misalignment-induced intensity term as
\begin{equation}\label{equ:A2A_I}
\mathcal{I}_{m,n}^{i,j} \triangleq ( u^{(x)i,j}_{m,n} )^2 + ( \sin\alpha \, u^{(y)i,j}_{m,n} - \cos\alpha \, u^{(z)i,j}_{m,n} )^2.
\end{equation}

\begin{figure*}[!t]
\begin{equation}\label{equ:A2A_I_2}
      \mathcal{I} _{m,n}^{i,j}=(d_m-d_i\cos \phi -d_j\sin \phi \sin \theta )^2+(d_n\cos \alpha +d_i\sin \alpha \sin \phi -d_j\left( \cos \theta \cos \alpha +\cos \phi \sin \theta \sin \alpha \right) )^2.
\end{equation}
\hrulefill
\end{figure*}

Expanding $\mathcal{I}_{m,n}^{i,j}$ yields the expression shown in~\eqref{equ:A2A_I_2} at the top of the this page. To evaluate $\max_{m,n,i,j} \mathcal{I}_{m,n}^{i,j}$, we note that $\mathcal{I}_{m,n}^{i,j}$ is convex with respect to $(d_i, d_j, d_m, d_n)$ over the compact set $[-D_1/2, D_1/2]^2 \times [-D_2/2, D_2/2]^2$. The maximum is thus attained at the vertices. Since $d_m$ and $d_n$ are uncoupled in~\eqref{equ:A2A_I_2}, we set $d_m = -\text{sgn}(d_i\cos \phi +d_j\sin \phi \sin \theta)D_2/2$ and $d_n = -\text{sgn}(d_i\sin \alpha \sin \phi -d_j\left( \cos \theta \cos \alpha +\cos \phi \sin \theta \sin \alpha \right))D_2/2$. Substituting these into~\eqref{equ:define_near-field_U2U_offboresight}, \textbf{the corresponding P2P off-boresight near-field distance (with dual-angle rotation) is given by}
\begin{equation}\label{equ:A2A_offboresight}
    r_{\text{F,off}}^{\text{P2P}}(\theta,\phi,\alpha)=\max \left( r_{\text{(a)}}^{\text{P2P}}(\theta,\phi,\alpha),r_{\text{(b)}}^{\text{P2P}}(\theta,\phi,\alpha) \right),
\end{equation}
where $r_{\text{(a)}}^{\text{P2P}}(\theta,\phi,\alpha)$ and $r_{\text{(b)}}^{\text{P2P}}(\theta,\phi,\alpha)$ are defined in~\eqref{equ:UPA-to-UPA_dual-angle_a} and~\eqref{equ:UPA-to-UPA_dual-angle_b}, respectively, which are located at the top of the next page. These expressions are attained under the conditions $d_i = \operatorname{sgn}(d_j)$ and $d_i = -\operatorname{sgn}(d_j)$, respectively.

\begin{figure*}[!t]
\begin{equation}\label{equ:UPA-to-UPA_dual-angle_a}
     r_{\text{(a)}}^{\text{P2P}}(\theta,\phi,\alpha)=\frac{\pi (D_2+D_1|\cos \phi + \sin \phi \sin \theta |)^2}{4\lambda \varphi}+\frac{\pi (D_2\cos \alpha +D_1|\cos \theta \cos \alpha +\cos \phi \sin \theta \sin \alpha - \sin \alpha \sin \phi |)^2}{4\lambda \varphi}.
\end{equation}
\end{figure*}
\begin{figure*}[!t]
\begin{equation}\label{equ:UPA-to-UPA_dual-angle_b}
     r_{\text{(b)}}^{\text{P2P}}(\theta,\phi,\alpha)=\frac{\pi (D_2+D_1|\cos \phi - \sin \phi \sin \theta |)^2}{4\lambda \varphi}+\frac{\pi (D_2\cos \alpha +D_1|\cos \theta \cos \alpha +\cos \phi \sin \theta \sin \alpha + \sin \alpha \sin \phi |)^2}{4\lambda \varphi}.
\end{equation}
\hrulefill
\end{figure*}

\begin{figure*}[!t]
\centering
\subfloat[$\alpha = 0^{\circ}$]{\includegraphics[width=0.34\textwidth]{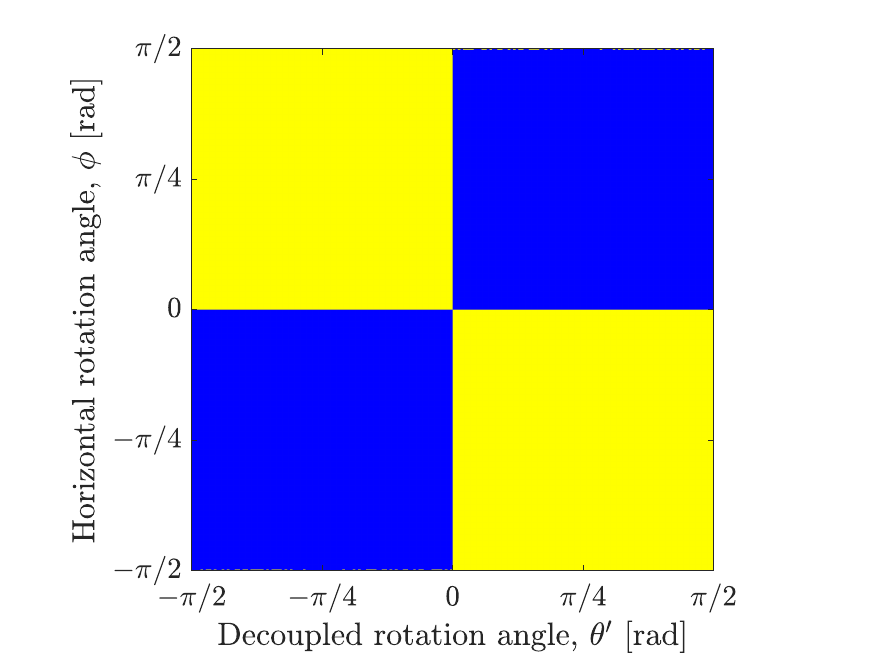}%
\label{fig:dominance_beta_0}}
\hspace{-3mm}
\subfloat[$\alpha = 30^{\circ}$]{\includegraphics[width=0.34\textwidth]{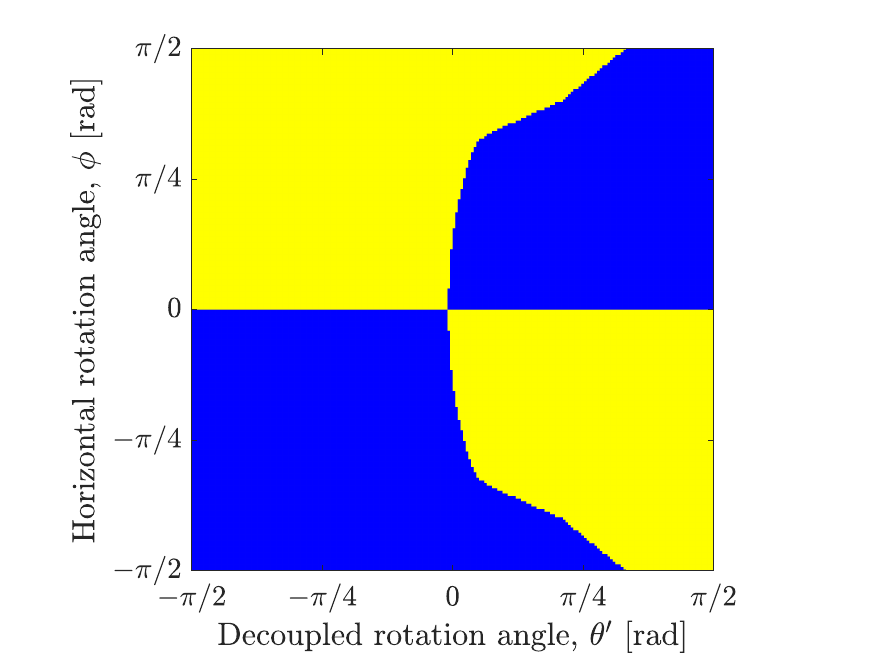}%
\label{fig:dominance_beta_30}}
\hspace{-3mm}
\subfloat[$\alpha = 60^{\circ}$]{\includegraphics[width=0.34\textwidth]{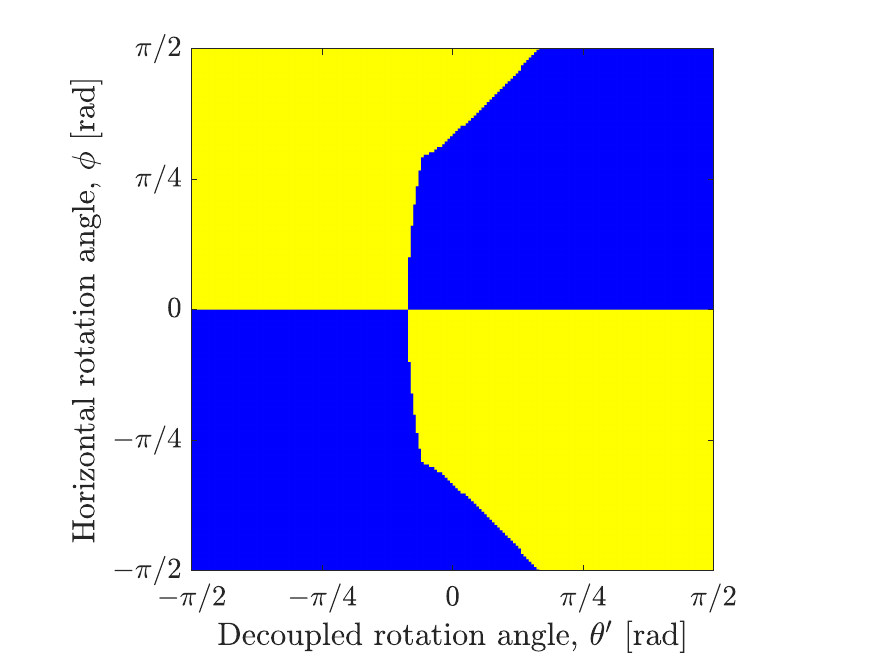}%
\label{fig:dominance_beta_60}}
\caption{Dominance regions of the two closed-form expressions in the $(\theta', \phi)$ parameter space for different off-boresight angles $\alpha$. The blue regions indicate where~\eqref{equ:UPA-to-UPA_dual-angle_a} yields a larger near-field distance, while the yellow regions correspond to the dominance of~\eqref{equ:UPA-to-UPA_dual-angle_b}.}
\label{fig:xi_dominance_regions}
\end{figure*}

\begin{figure}[!t]
    \centering
    \subfloat[No rotation.]{%
    \includegraphics[width=0.9\linewidth]{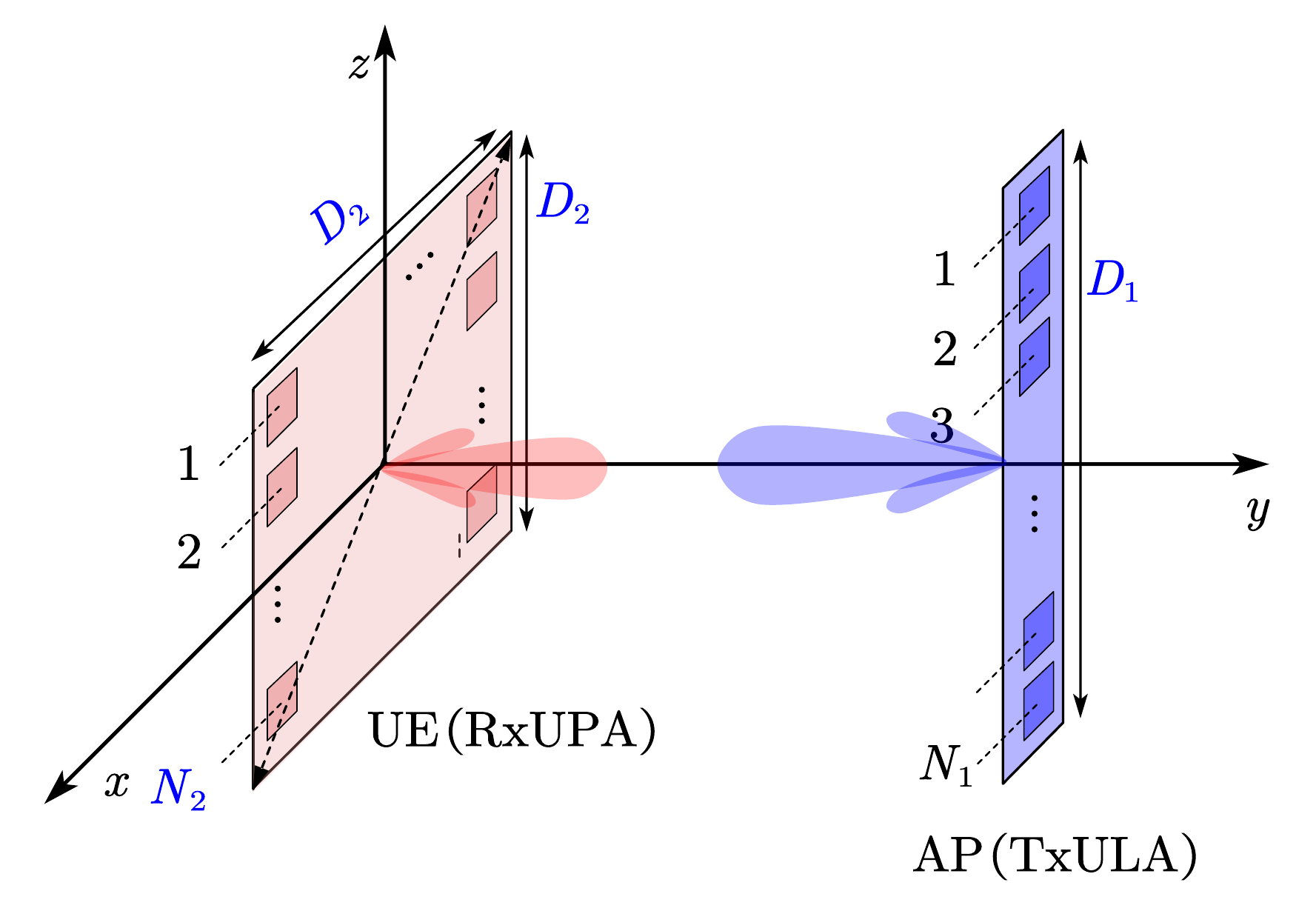}%
    }\hfill
    \subfloat[Single $x$-axis rotation.]{%
    \includegraphics[width=0.9\linewidth]{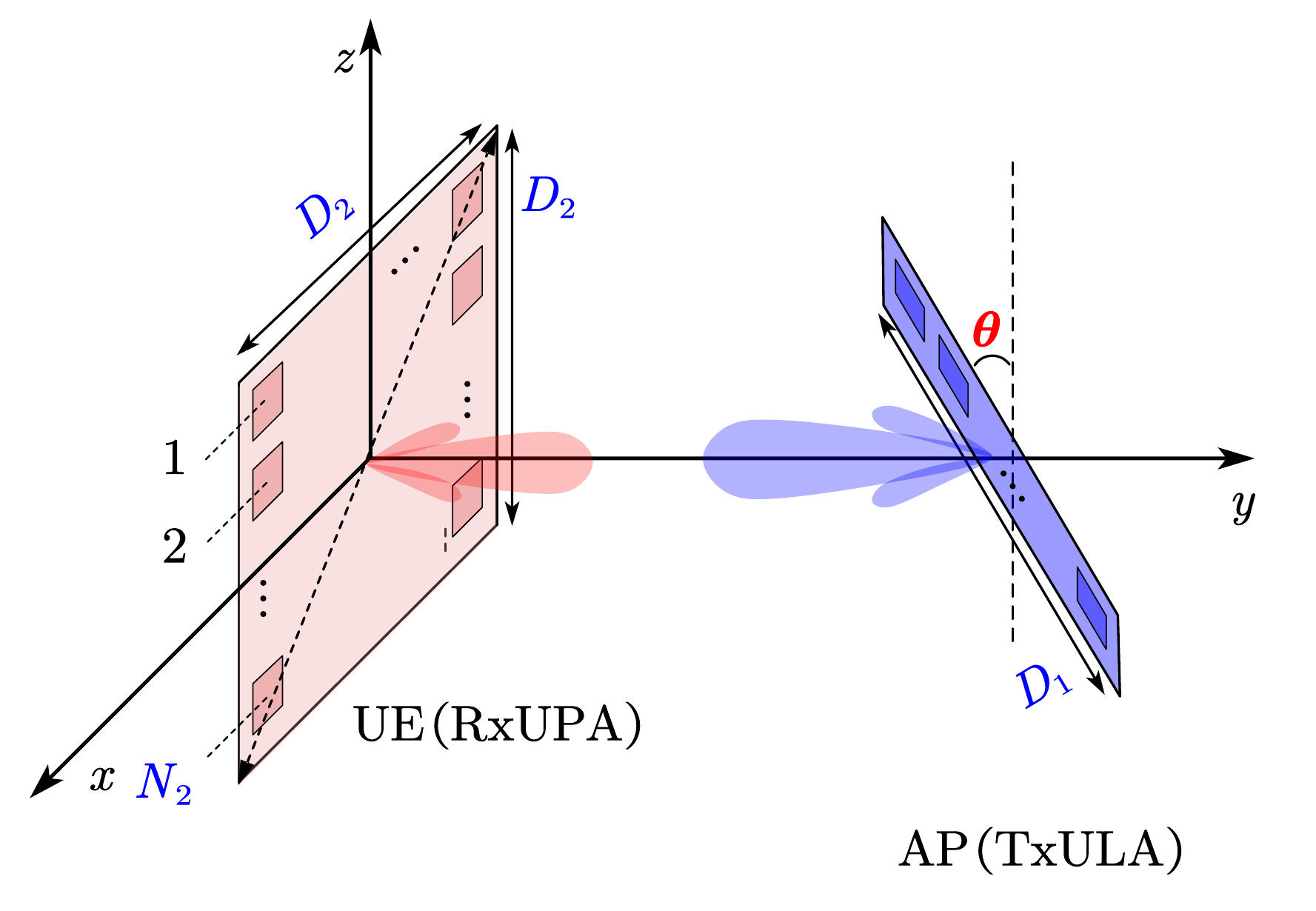}%
    }
    \caption{\rew{Mixed P2L on-boresight scenario.}}
    \label{fig:mixed_P2L_example}
\end{figure}

\rew{\textbf{Analysis:} The dual-angle off-boresight boundary in \eqref{equ:A2A_offboresight} takes a max over two branches \eqref{equ:UPA-to-UPA_dual-angle_a}-\eqref{equ:UPA-to-UPA_dual-angle_b} because different corner-element sign patterns can dominate the worst-case phase mismatch under general 3D orientations. Compared to the ULA case (where a single dominant expression typically suffices), this branch-switching behavior is a distinctive property of UPAs and directly reflects strong azimuth-elevation coupling induced by $(\theta,\phi)$ and the off-boresight displacement $\alpha$. The structure also yields consistent special cases: setting $\alpha=0$ reduces \eqref{equ:A2A_offboresight} to the on-boresight expression \eqref{equ:UPA-to-UPA_onboresight}, while setting $\phi=0$ collapses the dual-angle model to the single-angle off-boresight configuration in \eqref{equ:UPA-to-UPA_off-boresight_0}. The dominance maps in Fig.~\ref{fig:xi_dominance_regions} further visualize how the max in \eqref{equ:A2A_offboresight} induces asymmetric regions and nontrivial coupling effects as $\alpha$ increases, which cannot be captured by conventional no-rotation baselines.}

\subsection{\rew{Representative mixed P2L examples (on-boresight)}}\label{sec:mixed_P2L_on}

\rew{So far, we have focused on the two canonical and most representative aperture pairings, namely L2L and P2P, to keep a clear main storyline and to obtain compact closed-form characterizations. Nevertheless, the proposed phase-error-based methodology is geometry-driven and can be readily applied to other array-type combinations by specifying the corresponding element coordinates and orientation model. To illustrate this extensibility in a simple setting, we consider a mixed on-boresight P2L link, where the Rx is a square UPA (side length $D_2$) fixed in the $xz$-plane (boresight along $+y$), and the Tx is a ULA of length $D_1$ centered on the $y$-axis. We provide two representative cases: (i) no rotation, and (ii) a single $x$-axis rotation of the TxULA. A complete treatment of arbitrary mixed-array rotations is possible within the same framework, but it admits multiple non-equivalent geometries and may lead to lengthy case-by-case dominance analysis. Therefore, we use the following two cases as illustrative references.}

\rew{The RxUPA is fixed in the $xz$-plane and the $(m,n)$-th Rx element is located at $\mathbf{d}_{m,n}=(d_m,0,d_n)^\top$, where $d_m,d_n\in[-D_2/2,D_2/2]$. The TxULA has elements indexed by $n_1$ with local coordinate $\mathbf{d}_{n_1}=(0,0,d_{n_1})^\top$, where $d_{n_1}\in[-D_1/2,D_1/2]$. The center-to-center separation is $r$ along the $y$-axis, i.e., the Tx array center is at $r\mathbf{e}_y$ with $\mathbf{e}_y=(0,1,0)^\top$.}

\subsubsection*{\rew{Case~1: No rotation (Fig.~\ref{fig:mixed_P2L_example}(a))}}
\rew{The global position of the $n_1$-th Tx element is $\mathbf{d}'_{n_1}=r\mathbf{e}_y+\mathbf{d}_{n_1}$. The distance between the $(m,n)$-th Rx element and the $n_1$-th Tx element is}
\rew{\begin{equation}
r^{n_1}_{m,n}=\bigl\|\mathbf{d}_{m,n}-\mathbf{d}'_{n_1}\bigr\|
=\sqrt{r^2+d_m^2+(d_n-d_{n_1})^2}.
\end{equation}}
\rew{Applying the same third-order Taylor approximation as in Sec.~\ref{sec:ULA-to-ULA(on-boresight)}, we obtain}
\rew{\begin{equation}
r^{n_1}_{m,n}\approx r+\frac{d_m^2+(d_n-d_{n_1})^2}{2r}.
\end{equation}}
\rew{Hence, under the generalized phase threshold $\varphi$, the near-field condition becomes}
\rew{\begin{equation}
\max_{m,n,n_1}\frac{d_m^2+(d_n-d_{n_1})^2}{2r}\le \frac{\lambda\varphi}{2\pi}.
\end{equation}}
\rew{The maximum is attained at the aperture endpoints $|d_m|=D_2/2$ and $|d_n-d_{n_1}|=(D_2+D_1)/2$, which yields}
\rew{\begin{equation}
\frac{D_2^2+(D_1+D_2)^2}{8r}=\frac{\lambda\varphi}{2\pi}.
\end{equation}}
\rew{Solving for $r$, \textbf{the on-boresight near-field boundary distance for the mixed P2L link without rotation is}}
\rew{\begin{equation}
r_{\text{F,on}}^{\text{P2L}}(0)=\frac{\pi D_2^2}{4\lambda\varphi}+\frac{\pi (D_1+D_2)^2}{4\lambda\varphi}.
\label{equ:P2L_on_no_rot}
\end{equation}}

\subsubsection*{\rew{Case~2: Single-angle rotation of the TxULA around the $x$-axis (Fig.~\ref{fig:mixed_P2L_example}(b))}}
\rew{Now the TxULA is rotated by $\theta$ around the $x$-axis, so the rotated local coordinate becomes $\tilde{\mathbf{d}}_{n_1}=\mathbf{R}_x(\theta)\mathbf{d}_{n_1}=(0,-d_{n_1}\sin\theta, d_{n_1}\cos\theta)^\top$, where $\mathbf{R}_x(\theta)$ is given in~\eqref{rotation_x}. The global position of the $n_1$-th Tx element is $\mathbf{d}'_{n_1}=r\mathbf{e}_y+\tilde{\mathbf{d}}_{n_1}$. Define $\mathbf{v}^{n_1}_{m,n}=\mathbf{d}_{m,n}-\tilde{\mathbf{d}}_{n_1}=(d_m,\,d_{n_1}\sin\theta,\,d_n-d_{n_1}\cos\theta)^\top$. Then}
\rew{\begin{equation}
r^{n_1}_{m,n}=\sqrt{r^2-2r\,\mathbf{e}_y^\top\mathbf{v}^{n_1}_{m,n}+\|\mathbf{v}^{n_1}_{m,n}\|^2}.
\end{equation}}
\rew{Applying the same Taylor expansion as in Sec.~\ref{sec:ULA-to-ULA(on-boresight)} yields}
\rew{\begin{equation}
r^{n_1}_{m,n}\approx r-d_{n_1}\sin\theta+\frac{d_m^2+(d_n-d_{n_1}\cos\theta)^2}{2r}.
\end{equation}}
\rew{To steer the beam toward the Rx center along the $-y$ direction, the Tx applies a pre-steering phase that cancels the linear term in $d_{n_1}$, which is equivalent to using the effective distance $r'^{n_1}_{m,n}=r^{n_1}_{m,n}+d_{n_1}\sin\theta$. Therefore,}
\rew{\begin{equation}
r'^{n_1}_{m,n}\approx r+\frac{d_m^2+(d_n-d_{n_1}\cos\theta)^2}{2r}.
\end{equation}}
\rew{The near-field condition then becomes}
\rew{\begin{equation}
\max_{m,n,n_1}\frac{d_m^2+(d_n-d_{n_1}\cos\theta)^2}{2r}\le \frac{\lambda\varphi}{2\pi}.
\end{equation}}
\rew{The maximum is attained at $|d_m|=D_2/2$ and $|d_n-d_{n_1}\cos\theta|=(D_2+D_1\cos\theta)/2$, yielding}
\rew{\begin{equation}
\frac{D_2^2+(D_1\cos\theta+D_2)^2}{8r}=\frac{\lambda\varphi}{2\pi}.
\end{equation}}
\rew{Solving for $r$, \textbf{the on-boresight near-field boundary distance for the mixed P2L with a single $x$-axis rotation is}}
\rew{\begin{equation}
r_{\text{F,on}}^{\text{P2L}}(\theta)=\frac{\pi D_2^2}{4\lambda\varphi}+\frac{\pi (D_1\cos\theta+D_2)^2}{4\lambda\varphi}.
\label{equ:P2L_on_rot_x}
\end{equation}}

\rew{\textbf{Analysis:} Comparing \eqref{equ:P2L_on_rot_x} with the on-boresight L2L boundary in~\eqref{equ:U-to-ULA_on_boresight}, we observe a simple decomposition
$r_{\mathrm{F,on}}^{\mathrm{P2L}}(\theta)=\frac{\pi D_2^2}{4\lambda\varphi}+r_{\mathrm{F,on}}^{\mathrm{L2L}}(\theta)$,
where the additional term $\pi D_2^2/(4\lambda\varphi)$ stems from the extra transverse aperture dimension of the RxUPA (the $x$-dimension), while the second term follows the same projected-aperture structure driven by $D_1\cos\theta+D_2$.}

\rew{\textbf{Extensibility to other geometries:}
The above mixed P2L example demonstrates that the proposed derivation is purely geometry-driven: once the element coordinates and the rotation/orientation model are specified, the per-component effective phases (propagation plus Tx/Rx phase compensations) can be evaluated and the same phase-deviation criterion applies. Hence, beyond mixed array types, the framework can also accommodate non-coplanar (3D) orientations of ULA/UPA arrays by using the corresponding 3D rotations.}

\subsection{UPA-to-point case (off-boresight)}\label{sec:UPA-to-point}
\begin{figure}[!t]
    \centering
    \includegraphics[width=0.9\linewidth]{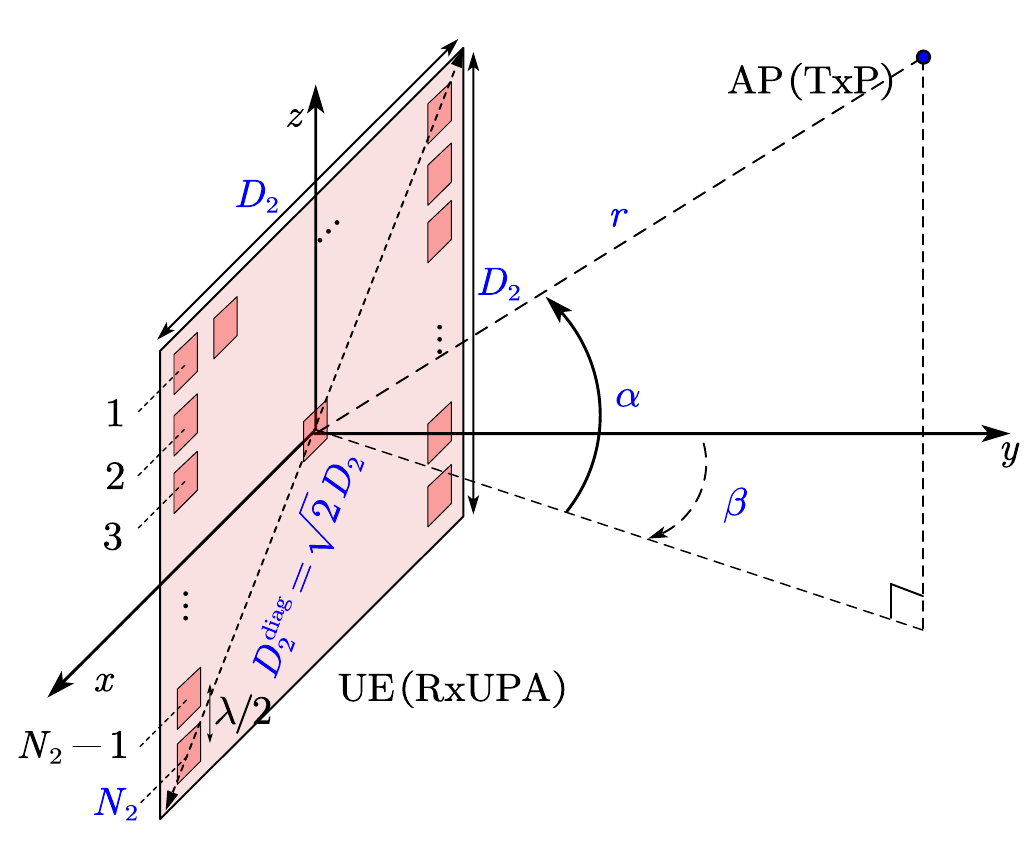}
    \caption{Off-boresight UPA-to-point setup (analyzed in Sec.~\ref{sec:UPA-to-point}).}
    \label{fig:UPA-to-point_system_model}
\end{figure}
We next consider a special case where the transmitter is reduced to a TxP, as in Fig.~\ref{fig:UPA-to-point_system_model}. The RxUPA is placed in the $xz$-plane and centered at the origin, while the point transmitter is located at a distance $r$ from the array center. Its position is characterized by two angular parameters, $\alpha$ and $\beta$, both ranging from $[-\pi/2, \pi/2]$. Specifically, $\alpha$ denotes the elevation angle (i.e., the deviation from the $xy$-plane toward the $z$-axis), and $\beta$ denotes the azimuth angle (i.e., the deviation from the $y$-axis toward the $x$-axis) in the standard spherical coordinate system. Under this definition, the Cartesian coordinates of the point source are given by $\mathbf{p} = (r \sin\beta \cos\alpha,\ r \cos\beta \cos\alpha,\ r \sin\alpha)^\top$.

\rew{Following the near-field criterion in the previous section, we define the off-boresight UPA-to-point near-field distance $r_{\text{F,off}}^{\text{P2O}}(\alpha,\beta)$ as}
\begin{equation}\label{equ:def_P2O_off}
\rew{r=r_{\text{F,off}}^{\text{P2O}}(\alpha,\beta),
\quad \text{s.t.}\quad
\max_{m,n} r'_{m,n}
-\min_{m,n} r'_{m,n}
\le \frac{\lambda\varphi}{2\pi}.}
\end{equation}

The physical distance between the $(i,j)$-th RxUPA element and the point transmitter is $r_{m,n}=\left\| \mathbf{d}_{m,n}-\mathbf{p} \right\|.$ To coherently combine the received signal from the TxP, each RxUPA element applies a pre-steering phase given by $\tilde{\phi}_{m,n} = -{2\pi} \mathbf{d}_{m,n}^\top \cdot \mathbf{u}_{\text{p}}/{\lambda}$ where $\mathbf{u}_{\mathrm{p}}= ( \sin\beta\cos\alpha, \cos\beta \cos\alpha, \sin\alpha)^{\top}$ denotes the unit vector pointing from the array center toward the TxP. Therefore, the total adjusted phase delay corresponds to an effective distance $r'_{m,n} = r_{m,n} + \mathbf{d}_{m,n}^\top \cdot \mathbf{u}_{\text{p}}.$

Note that the position of the point can be rewritten as $\mathbf{p}=r\mathbf{u}_{\text{p}}$. Applying the third-order Taylor approximation
\begin{equation}
r_{m,n}  = r - \mathbf{d}_{m,n}^\top \cdot \mathbf{u}_{\text{p}} + \frac{\|\mathbf{d}_{m,n}\|^2 - (\mathbf{d}_{m,n}^\top \cdot \mathbf{u}_{\text{p}})^2}{2r} + o\left(\frac{1}{r^2}\right).
\end{equation}

Substituting this approximation into the effective distance expression, the first-order linear term cancels out, leading to
\begin{equation}\label{equ:A2P_r}
r'_{m,n} \approx r + \frac{\|\mathbf{d}_{m,n}\|^2 - (\mathbf{d}_{m,n}^\top \cdot \mathbf{u}_{\text{p}})^2}{2r}.
\end{equation}

Note that the minimum value of $r'_{m,n}$ is achieved when $d_m=d_n=0$, i.e., $\min_{m,n}r'_{m,n}=r$. Substituting $\mathbf{u}_p$ into~\eqref{equ:A2P_r}, the near-field condition in previous section becomes
\begin{equation}
\max_{m,n} \frac{ d_{m}^{2}\cos ^2\beta +\left( d_m\sin \alpha \sin \beta -d_n\cos \alpha \right)  }{2r} \leq \frac{\lambda\varphi}{2\pi},
\end{equation}
which implies that the maximum deviation is attained at $d_m = \operatorname{sgn}(\sin \alpha \sin \beta)D_2/2$ and $d_n = -\operatorname{sgn}(d_m)D_2/2$. Therefore, \textbf{the corresponding UPA-to-point off-boresight near-field distance is given by}
\begin{equation}
\label{equ:UPA-to-point_off-boresight}
r_{\text{F,off}}^{\text{P2O}}(\alpha,\beta) = \frac{\pi D_{2}^{2}\cos ^2\beta}{4\lambda \varphi}+\frac{\pi D_{2}^{2}\left( |\sin \alpha \sin \beta |+\cos \alpha \right) ^2}{4\lambda \varphi}.
\end{equation}

\rew{\textbf{Analysis:} The UPA-to-point boundary in \eqref{equ:UPA-to-point_off-boresight} again appears as the sum of two Fraunhofer-type terms, corresponding to the two transverse dimensions of the UPA aperture. The first term scales with $\cos^2\beta$ and captures the azimuth-induced projection in one dimension, while the second term scales with $(|\sin\alpha\sin\beta|+\cos\alpha)^2$ and reveals explicit azimuth--elevation coupling through $(\alpha,\beta)$. This coupling implies that ignoring the angular dependence can misestimate near-field distance under general off-boresight locations, even for a point transmitter.} When $\varphi = \pi/8$, the expression in \eqref{equ:UPA-to-point_off-boresight} reduces to the generalized Fraunhofer distance criterion
\begin{equation}
r_{\text{F,off}}^{\text{P2O}}(\alpha,\beta)|_{\varphi=\frac{\pi}{8}} = \frac{2 D_{2}^{2}}{\lambda} \left[ \cos ^2\beta + \left( |\sin \alpha \sin \beta |+\cos \alpha \right) ^2 \right].
\end{equation}

Furthermore, for the on-boresight case where $\alpha = 0$ and $\beta = 0$ (i.e., the point of interest is aligned with the TxUPA's boresight), the expression simplifies to
\begin{equation}
r_{\text{F,off}}^{\text{P2O}}(0,0)|_{\varphi=\frac{\pi}{8}} =  \frac{\pi D_{2}^{2}}{2\lambda \varphi}.
\end{equation}

\rew{To facilitate readability and practical reuse, we summarize in Table~\ref{tab:nf_summary} the key closed-form near-field boundaries derived in Sections~III-IV across all considered geometries and misalignment settings.}

\begin{table*}[!t]
\centering
\caption{\rew{Main derived closed-form near-field boundary distances under the considered geometries and misalignment settings.}}
\label{tab:nf_summary}
\footnotesize
\setlength{\tabcolsep}{4pt}
\renewcommand{\arraystretch}{1.32}
\begin{tabular}{@{}
  p{0.15\textwidth}
  p{0.18\textwidth}
  p{0.59\textwidth}
  p{0.04\textwidth}
@{}}

\toprule  
\textbf{Geometry} & \textbf{Setting} & \textbf{Closed-form $r_{\mathrm{F}}$} & \textbf{Ref.}\\
\specialrule{1pt}{0pt}{2pt}  

\multicolumn{4}{@{}l}{\textit{\textbf{ULA-based (1D aperture)}}}\\[0.5mm]
\specialrule{1pt}{0pt}{2pt}  

ULA-to-ULA (L2L) & On-boresight &
$\displaystyle
  r_{\text{F,on}}^{\text{L2L}}(\theta)
  = \frac{\pi\!\left(D_1\cos\theta+D_2\right)^2}{4\lambda\varphi}
$ &
\eqref{equ:U-to-ULA_on_boresight}\\

\specialrule{0.25pt}{2pt}{2pt}  

\multirow{3}{=}{ULA-to-ULA (L2L)} &
\multirow{3}{=}{Off-boresight} &
$\displaystyle
  r_{\text{F,off}}^{\text{L2L}}(\theta,\alpha)
  = \max\!\Big(r_{(\text{a})}^{\text{L2L}}(\theta,\alpha),\;
               r_{(\text{b})}^{\text{L2L}}(\theta,\alpha)\Big)
$ &
\multirow{3}{=}{\eqref{equ:L2L_exact}}\\[1.5mm]
& &
$\displaystyle
  r_{(\text{a})}^{\text{L2L}}
  = \frac{\pi\!\left(D_1\cos(\theta-\alpha)+D_2\cos\alpha\right)^2}{4\lambda\varphi}
    +\frac{1}{2}\Big|D_1\sin(\theta-\alpha)-D_2\sin\alpha\Big|
$ & \\[3.5mm]
& &
$\displaystyle
  r_{(\text{b})}^{\text{L2L}}
  = \frac{\pi\!\left(D_1\cos(\theta-\alpha)-D_2\cos\alpha\right)^2}{4\lambda\varphi}
    +\frac{1}{2}\Big|D_1\sin(\theta-\alpha)+D_2\sin\alpha\Big|
$ & \\[2mm]

\specialrule{0.25pt}{2pt}{2pt}  

ULA-to-point (L2O) & Off-boresight &
$\displaystyle
  r_{\text{F,off}}^{\text{L2O}}(\alpha)
  = \frac{\pi D_2^2\cos^2\!\alpha}{4\lambda\varphi}
    + \frac{D_2}{2}\,|\sin\alpha|
$ &
\eqref{equ:L2O_off_closed}\\[1mm]

\specialrule{0.7pt}{3pt}{3pt}   

\multicolumn{4}{@{}l}{\textit{\textbf{UPA-based (2D aperture) and representative mixed P2L examples (on-boresight)}}}\\[0.5mm]
\specialrule{0.8pt}{0pt}{2pt}  

UPA-to-UPA (P2P) & On-boresight &
$\displaystyle
  r_{\text{F,on}}^{\text{P2P}}(\theta,\phi)
  = \frac{\pi\left(D_1\cos\theta+D_2\right)^2}{4\lambda\varphi}
    + \frac{\pi\left(D_1\!\left(\cos\phi+|\sin\theta\sin\phi|\right)+D_2\right)^2}{4\lambda\varphi}
$ &
\eqref{equ:UPA-to-UPA_onboresight}\\[2mm]

\specialrule{0.25pt}{2pt}{2pt}  

UPA-to-UPA (P2P) & \makecell[tl]{Off-boresight\\(single angle)} &
$\displaystyle
  r_{\text{F,off}}^{\text{P2P}}(\theta,\alpha)
  = \frac{\pi\left(D_1+D_2\right)^2}{4\lambda\varphi}
    + \frac{\pi\left(D_1\cos(\theta-\alpha)+D_2\cos\alpha\right)^2}{4\lambda\varphi}
$ &
\eqref{equ:UPA-to-UPA_off-boresight_0}\\[2mm]

\specialrule{0.25pt}{2pt}{2pt}

\multirow{3}{=}{UPA-to-UPA (P2P)} &
\multirow{3}{=}{\makecell[tl]{Off-boresight\\(dual angles)}} &
$\displaystyle
  r_{\text{F,off}}^{\text{P2P}}(\theta,\phi,\alpha)
  = \max\!\Big(r_{(\text{a})}^{\text{P2P}},\;r_{(\text{b})}^{\text{P2P}}\Big)
$ &
\multirow{3}{=}{\eqref{equ:A2A_offboresight}}\\[1.5mm]
& &
$\displaystyle
  r_{(\text{a})}^{\text{P2P}}
  = \frac{\pi\left(D_2+D_1\eta_{+}\right)^2}{4\lambda\varphi}
    + \frac{\pi\left(D_2\cos\alpha+D_1\xi_{-}\right)^2}{4\lambda\varphi}
$ & \\[3.5mm]
& &
$\displaystyle
  r_{(\text{b})}^{\text{P2P}}
  = \frac{\pi\left(D_2+D_1\eta_{-}\right)^2}{4\lambda\varphi}
    + \frac{\pi\left(D_2\cos\alpha+D_1\xi_{+}\right)^2}{4\lambda\varphi}
$ & \\[2mm]

\specialrule{0.25pt}{0pt}{2pt}  

UPA-to-ULA (P2L) & No rotation &
$\displaystyle
  r_{\text{F,on}}^{\text{P2L}}(0)
  = \frac{\pi D_2^2}{4\lambda\varphi}
    + \frac{\pi\left(D_1+D_2\right)^2}{4\lambda\varphi}
$ &
\eqref{equ:P2L_on_no_rot}\\[2mm]

\specialrule{0.25pt}{2pt}{2pt}  

UPA-to-ULA (P2L) & Tx $x$-axis rotation $\theta$ &
$\displaystyle
  r_{\text{F,on}}^{\text{P2L}}(\theta)
  = \frac{\pi D_2^2}{4\lambda\varphi}
    + \frac{\pi\left(D_1\cos\theta+D_2\right)^2}{4\lambda\varphi}
$ &
\eqref{equ:P2L_on_rot_x}\\[1mm]

\specialrule{0.25pt}{2pt}{2pt}  

UPA-to-point (P2O) & Off-boresight &
$\displaystyle
  r_{\text{F,off}}^{\text{P2O}}(\alpha,\beta)
  = \frac{\pi D_{2}^{2}\cos^{2}\!\beta}{4\lambda\varphi}
    + \frac{\pi D_{2}^{2}\!\left(|\sin\alpha\sin\beta|+\cos\alpha\right)^{2}}{4\lambda\varphi}
$ &
\eqref{equ:UPA-to-point_off-boresight}\\[1mm]

\bottomrule  
\end{tabular}

\begin{minipage}{\textwidth}
\footnotesize
\emph{Notation:} 
For compactness, define
$\eta_{\pm}(\theta,\phi)\triangleq|\cos\phi\pm\sin\phi\sin\theta|$ and
$\xi_{\pm}(\theta,\phi,\alpha)\triangleq|\cos\theta\cos\alpha+\cos\phi\sin\theta\sin\alpha\pm\sin\alpha\sin\phi|$.
\end{minipage}
\end{table*}

\section{Numerical Results}\label{sec:results}
In this section, the key analytical expressions developed in Sec.~\ref{sec:ULA} and Sec.~\ref{sec:UPA} are numerically elaborated. By default (unless specified otherwise for the particular figure), we set the signal frequency to $f = 300$\,GHz (wavelength of $\lambda=1\,$mm). We also validate the accuracy of our closed-form expressions for $r_\mathrm{F}$ by cross-verifying the obtained results across those delivered by an in-house computer simulator developed in~\cite{zhang2025impact}. This simulation tool utilizes exhaustive search to approximate $r_\mathrm{F}$ numerically by following the definition of the near-field boundary: as the minimal separation distance $r$ at which the maximum phase mismatch between any two received signal components remains below the threshold $\varphi$. 

\rew{To quantify when neglecting rotations becomes inaccurate, we use a relative deviation metric between the near-field boundary distance predicted with and without rotations. Specifically, for a given configuration (ULA/UPA) and the corresponding closed-form expression, let $r_{\mathrm F}(\cdot)$ denote the boundary distance computed with nonzero rotation angles, and let $r_{\mathrm F}^{(0)}$ denote the baseline obtained by setting the rotation angles to zero in the same configuration. We then define the relative deviation as $\Delta_{r_{\mathrm F}}\triangleq|r_{\mathrm F}(\cdot)-r_{\mathrm F}^{(0)}|/r_{\mathrm F}^{(0)}$. Unless stated otherwise, the percentages reported in the numerical results (e.g., Figs.~\ref{fig:diff_RxSize} and~\ref{fig:nearfield_distribution_alpha}) correspond to $100\times\Delta_{r_{\mathrm F}}$.}

\begin{figure}[!t]
    \centering
    \includegraphics[width=\linewidth]{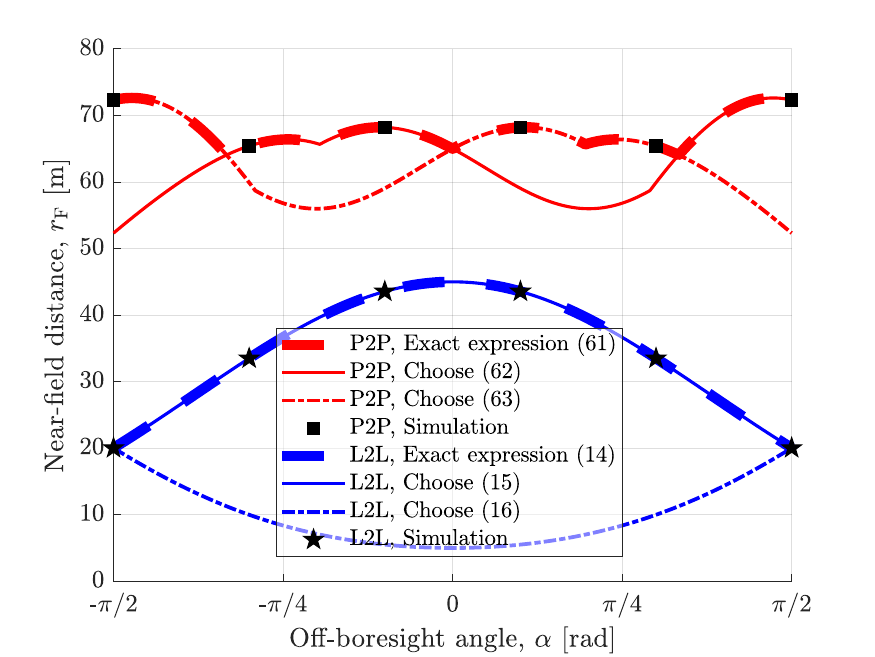}
    \caption{Comparison of the near-field distance expressions for P2P and L2L configurations as a function of the off-boresight angle $\alpha$ with $D_1 = 0.1$ m, $D_2 = 0.05$ m, $\theta'= 0$, $\phi=60^{\circ}$.}
    \label{fig:chooseExp}
\end{figure}

\subsection{Dominant expression selection}
We start by evaluating the accuracy and applicability of different closed-form expressions for the near-field distance under varying off-boresight angles. Fig.~\ref{fig:chooseExp} compares the analytical and simulated results for both P2P and L2L configurations as a function of the off-boresight angle $\alpha$, with $D_1 = 0.1$\,m, $D_2 = 0.05$\,m, $\theta' = 0$, and $\phi = 60^\circ$. As shown in the figure, the simulation results (square and star markers) closely match the theoretical expressions, which validate the tightness of the approximations. For the L2L case, the expression in~\eqref{equ:U2U_off_a} alone  accurately matches the exact result~\eqref{equ:L2L_exact} across the entire angular range. This suggests that a single dominant expression~\eqref{equ:U2U_off_a} is sufficient for this scenario. In contrast, the P2P configuration exhibits a more intricate angular dependence: neither~\eqref{equ:UPA-to-UPA_dual-angle_a} nor~\eqref{equ:UPA-to-UPA_dual-angle_b} alone can approximate the exact expression in~\eqref{equ:A2A_offboresight} accurately, \textbf{so both branches of the piecewise formulation are necessary for P2P scenarios.}

\begin{figure}[!t]
    \centering
    \includegraphics[width=\linewidth]{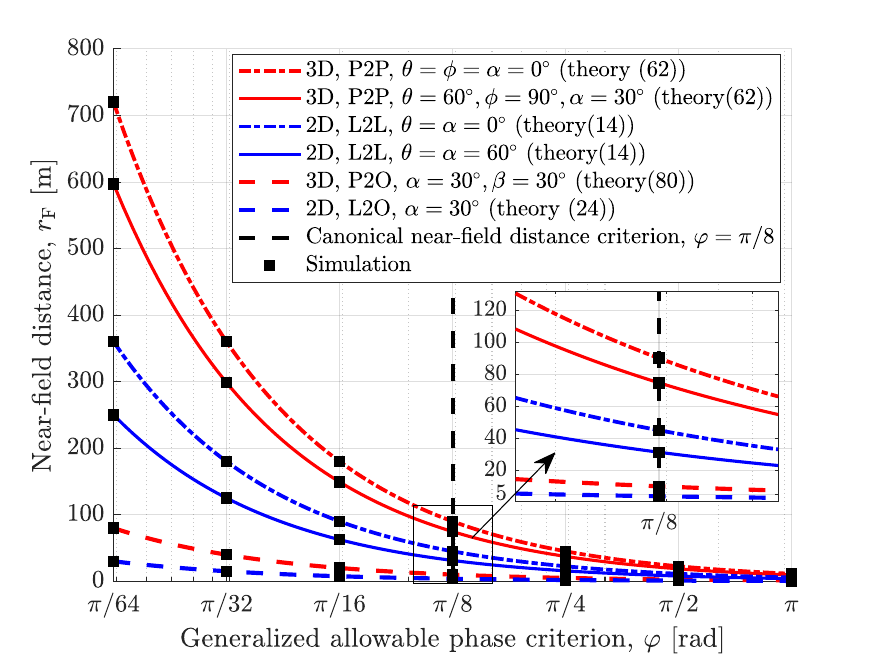}
    \caption{Near-field distance $r_{\text{F}}$ versus generalized allowable phase criterion $\varphi$ with $D_1=0.1\text{ m}$ and $D_2=0.05\text{ m}$.}
    \label{fig:sim_extended_criterion}
\end{figure}

\begin{figure*}[!t]
\centering
\subfloat[Fixed AP Size $D_1 = 0.1$ m]{\includegraphics[width=0.35\textwidth]{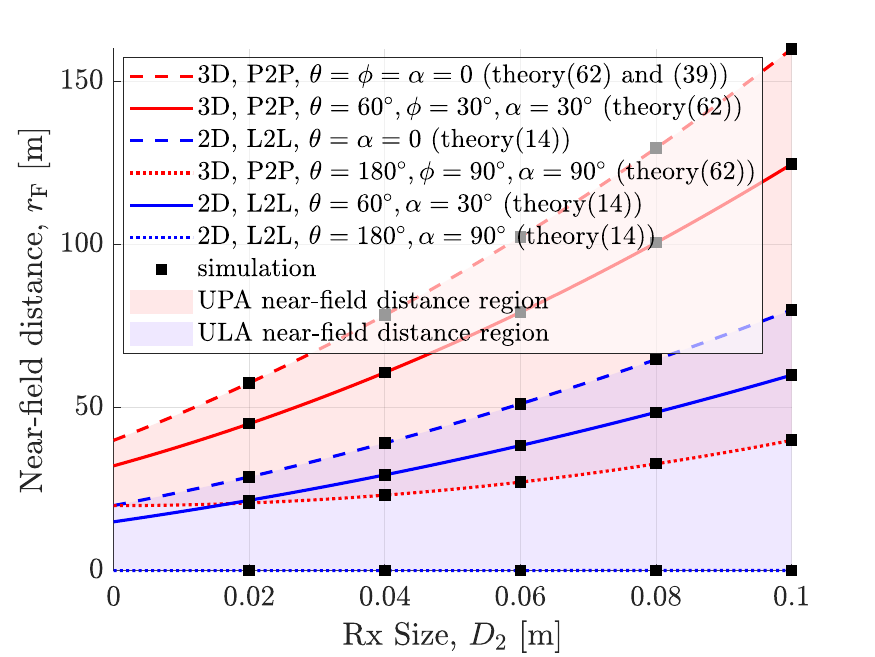}%
\label{fig:ula_upa_rx_size}}
\hspace{-0.6cm}
\subfloat[Near-field distance contours over $D_1$ and $D_2$]{\includegraphics[width=0.35\textwidth]{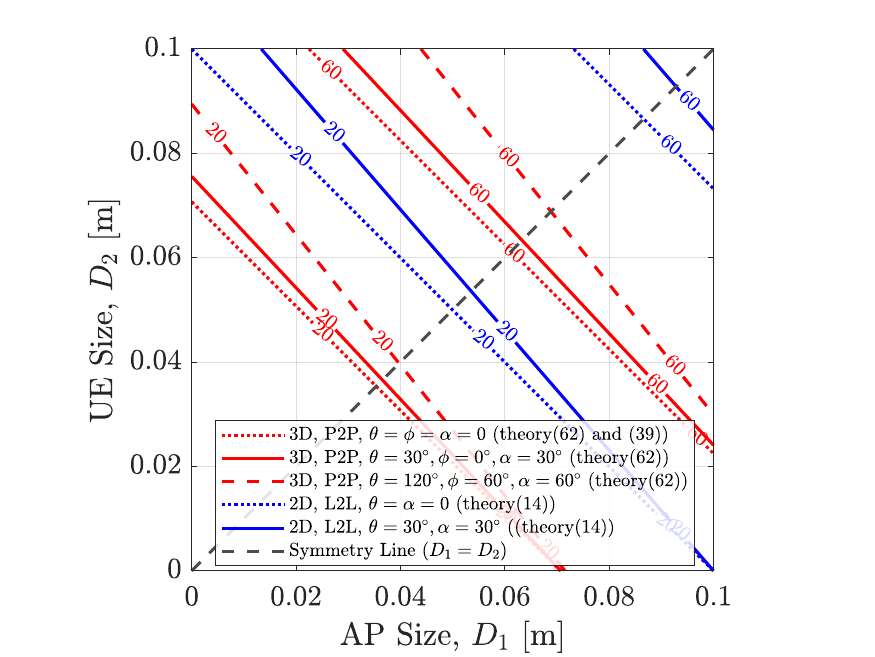}%
\label{fig:upa_rx_size}}
\hspace{-0.6cm}
\subfloat[Diagonal-matched comparison (fixed $D_2=0.1$ m)]{\includegraphics[width=0.35\textwidth]{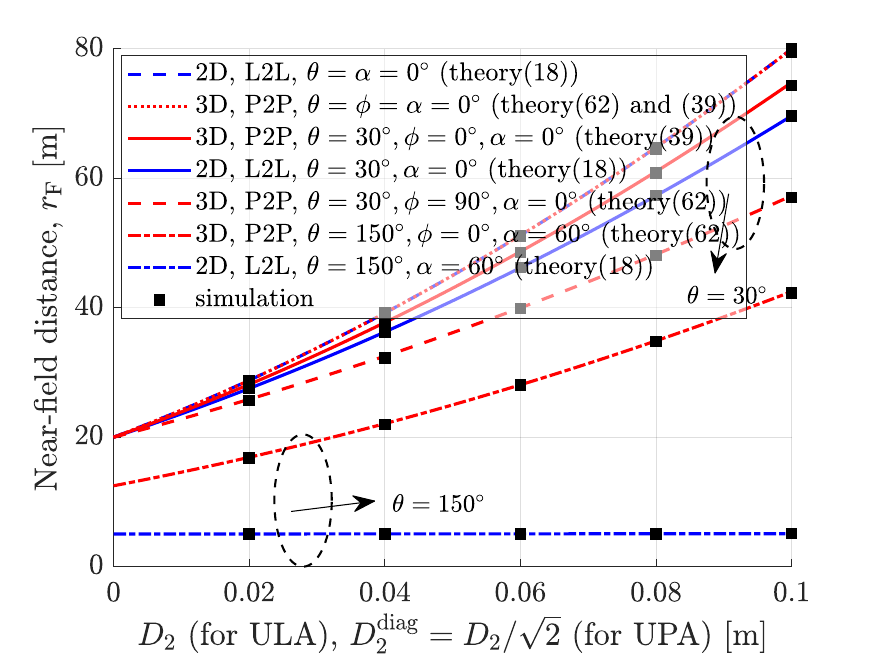}%
\label{fig:tx_size_comparison}}
\caption{Comprehensive comparison of near-field distance characteristics for different antenna array configurations and size variations.}
\label{fig:nearfield_comprehensive_comparison}
\end{figure*}

\subsection{Effect of general allowable phase criterion}
We now examine how the generalized phase mismatch threshold $\varphi$ influences the near-field distance $r_{\mathrm{F}}$. Fig.~\ref{fig:sim_extended_criterion} presents analytical and simulated results for various 2D and 3D configurations, including both P2P and L2L setups, as $\varphi$ varies from $\pi/64$ to $\pi$. The array apertures are fixed at $D_1=0.1$\,m and $D_2=0.05$\,m. The simulation results (black square markers) closely follow the theoretical predictions for all configurations, which validates the correctness of the generalized phase error criterion. As expected, smaller values of $\varphi$ lead to larger near-field distances $r_{\mathrm{F}}$, as a more stringent phase-error budget requires greater propagation distances to satisfy far-field conditions. Moreover, 3D P2P configurations yield notably larger near-field distances than their 2D counterparts for the same $\varphi$ due to the inherently larger aperture of the UPA. The difference becomes particularly visible under stringent phase thresholds (e.g., $\varphi \leq \pi/8$). \rew{Overall, introducing $\varphi$ provides a tunable generalization of the canonical $\pi/8$ Fraunhofer rule and offers a unified reference for different phase-accuracy budgets. Since the remainder of this section focuses on the impact of rotations/misalignment and aims at direct comparison with conventional Fraunhofer-type baselines, we use $\varphi=\pi/8$ as the default setting unless stated otherwise. } \rew{Operationally, under a prescribed phase-accuracy budget (e.g., $\varphi=\pi/8$), the resulting Fraunhofer-type boundary can be interpreted as the distance beyond which conventional range-independent far-field beamsteering becomes gain-admissible, whereas below it one may need range-dependent near-field focusing to avoid antenna gain loss~\cite{bjornson2021primer}.}

\subsection{Effect of AP and UE size}
We next investigate how the physical aperture sizes of the AP and UE affect the near-field distance $r_{\mathrm{F}}$. Fig.~\ref{fig:ula_upa_rx_size}--\ref{fig:tx_size_comparison} illustrate the impact of varying the array dimensions under different configurations and angular settings. Specifically, Fig.~\ref{fig:ula_upa_rx_size} compares ULA and UPA with the AP aperture fixed at $D_1=0.1$ m. In general, UPA yields larger $r_{\mathrm{F}}$ due to its angular sensitivity in both azimuth and elevation. However, for certain angular settings, the UPA near-field distance can fall below that of ULA, which reflects its additional degrees of freedom in misalignment. Moreover, unlike ULA, the UPA maintains a nonzero near-field distance even when $D_2 \to 0$ due to the planar structure.

To jointly examine the effect of both AP and UE apertures, Fig.~\ref{fig:upa_rx_size} shows contour plots of equal $r_{\mathrm{F}}$ in the $(D_1,D_2)$ domain. Without rotation, the contours are symmetric with respect to the line $D_1=D_2$, which indicates that the two dimensions are interchangeable. When rotation is introduced, this symmetry is broken and the two dimensions become non-interchangeable. As $r_{\mathrm{F}}$ increases, the ULA gradually exhibits weaker asymmetry, whereas the UPA continues to show strong asymmetry due to its additional angular degrees of freedom.

Fig.~\ref{fig:tx_size_comparison} next examines the effect of increasing the AP-side aperture with $D_2=0.1$\,m fixed. To enable an alternative comparison, we adopt a diagonal-matched baseline with $D_1^{\mathrm{UPA}} = {D_1}/{\sqrt{2}}$, $D_2^{\mathrm{UPA}} = {D_2}/{\sqrt{2}}$ so that the UPA diagonal length equals the ULA aperture length. Under this condition, ULA and UPA yield identical $r_{\mathrm{F}}$ without rotation. When a single rotation angle $\theta$ is applied, UPA achieves a larger near-field distance than ULA, consistent with the analytical proof in Appendix~\ref{appendix:UPA-ULA-comparison}. In contrast, under multi-angle settings (e.g., $\theta=30^\circ$ for both arrays while $\phi=90^\circ$ for UPA), UPA may instead yield smaller $r_{\mathrm{F}}$, which highlights the fundamental role of array geometry in determining $r_{\mathrm{F}}$. Across all configurations, the simulation results closely match the theory derived. For clarity, simulation points are omitted in the subsequent figures.

\begin{figure}[!t]
    \centering
    \includegraphics[width=\linewidth]{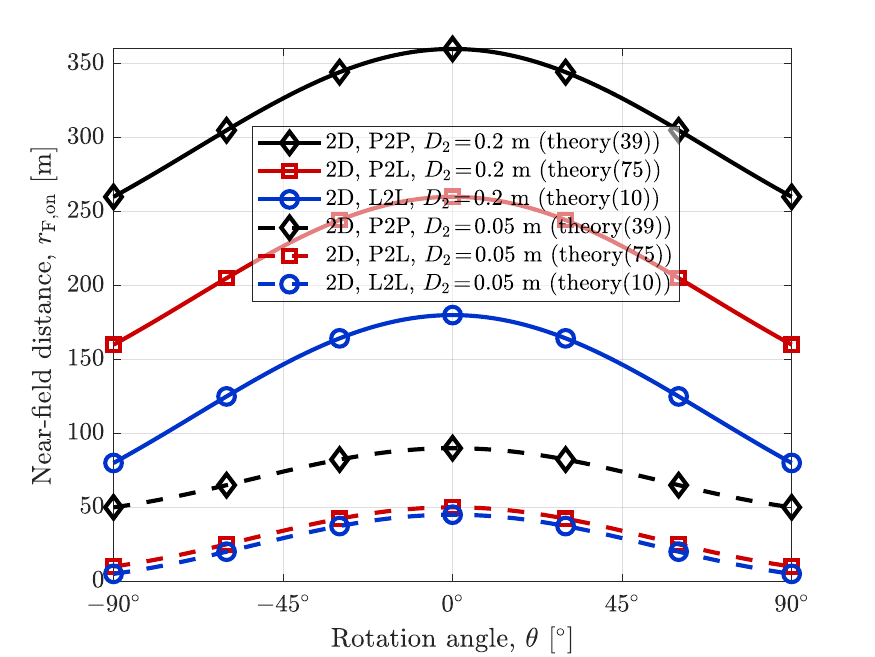}
    \caption{\rew{Comparison of on-boresight near-field boundary distances versus the rotation angle $\theta$ for L2L, P2L, and P2P, shown for two Rx aperture sizes $D_2\in\{0.05,0.2\}$~m.}}
    \label{fig:mixed_compare_D2}
\end{figure}

\rew{Fig.~\ref{fig:mixed_compare_D2} further compares the on-boresight near-field boundary $r_{\mathrm{F,on}}$ versus the rotation angle $\theta$ across three link types, namely L2L, the mixed P2L case, and P2P, while also varying the Rx-side aperture size $D_2$ (solid: $D_2=0.2$~m; dashed: $D_2=0.05$~m). Note that, for both $D_2$ values, the ordering $r_{\mathrm{F,on}}^{\mathrm{L2L}}(\theta) < r_{\mathrm{F,on}}^{\mathrm{P2L}}(\theta) < r_{\mathrm{F,on}}^{\mathrm{P2P}}(\theta)$ holds for all $\theta$, which confirms that introducing an additional transverse aperture dimension (UPA) systematically enlarges the near-field region relative to a purely 1D pairing.
Moreover, increasing $D_2$ enlarges all three boundaries, but the increase is more pronounced for P2L and P2P due to the extra $D_2^2$-type contribution brought by the planar aperture dimension.
Finally, all curves exhibit a symmetric decrease as $|\theta|$ increases, which reflects the projected-aperture reduction through the $\cos\theta$ factor.}

\begin{figure}[!t]
    \centering
    \includegraphics[width=\linewidth]{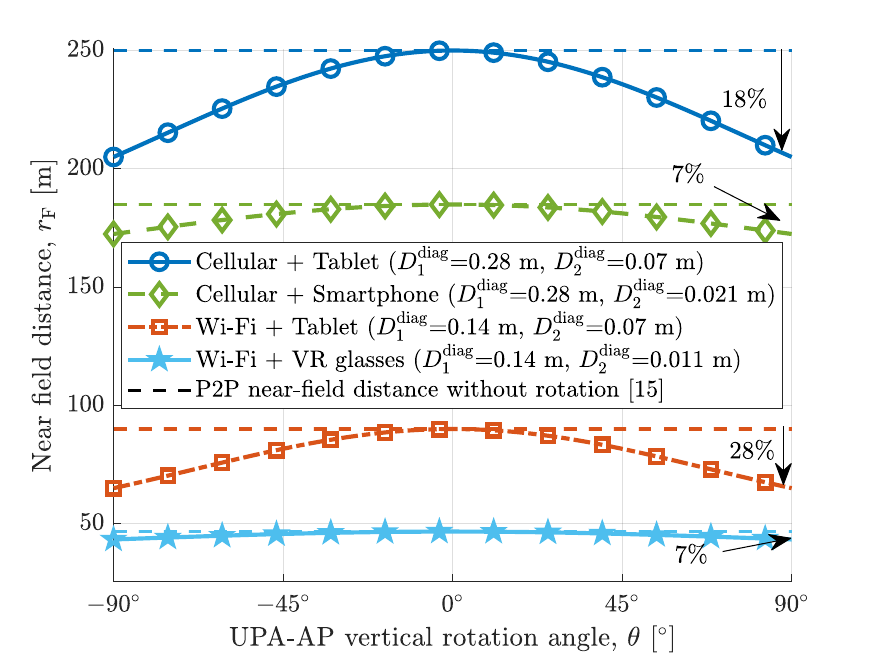}
    \caption{Near-field distance of P2P scenario in \eqref{equ:A2A_offboresight} versus UPA-AP vertical rotation angle $\theta$ with $\alpha=\phi=0^{\circ}$.}
    \label{fig:diff_RxSize}
\end{figure}

To further explore the effect of device dimensions and rotations, Fig.~\ref{fig:diff_RxSize} plots $r_{\mathrm{F}}$versus rotation angle $\theta$ for several realistic AP–UE combinations: large cellular ($D_1 = 0.2$\,m) and smaller Wi-Fi APs ($D_1 = 0.1$\,m); and UEs modeled as tablets ($D_2 = 0.05$\,m), smartphones ($D_2 = 0.015$\,m), and VR glasses ($D_2 = 0.008$\,m). The results show that $r_{\mathrm{F}}$ is symmetric with respect to $\theta = 0^\circ$ and generally decreases with increased rotation due to reduced effective aperture. Furthermore, larger UE arrays exhibit stronger variations in near-field distance. For example, under a cellular AP, the near-field distance changes by up to $18\%$ for a tablet, but only $7\%$ for a smartphone. Interestingly, for the same UE (tablet), a smaller AP (Wi-Fi) leads to a larger variation ($28\%$) than a larger one (cellular, $18\%$), due to the latter’s finer pre-steering phase capability.

\begin{figure*}[!t]
    \centering
    \begin{subfigure}[t]{0.51\linewidth}
        \centering
        \includegraphics[width=\linewidth]{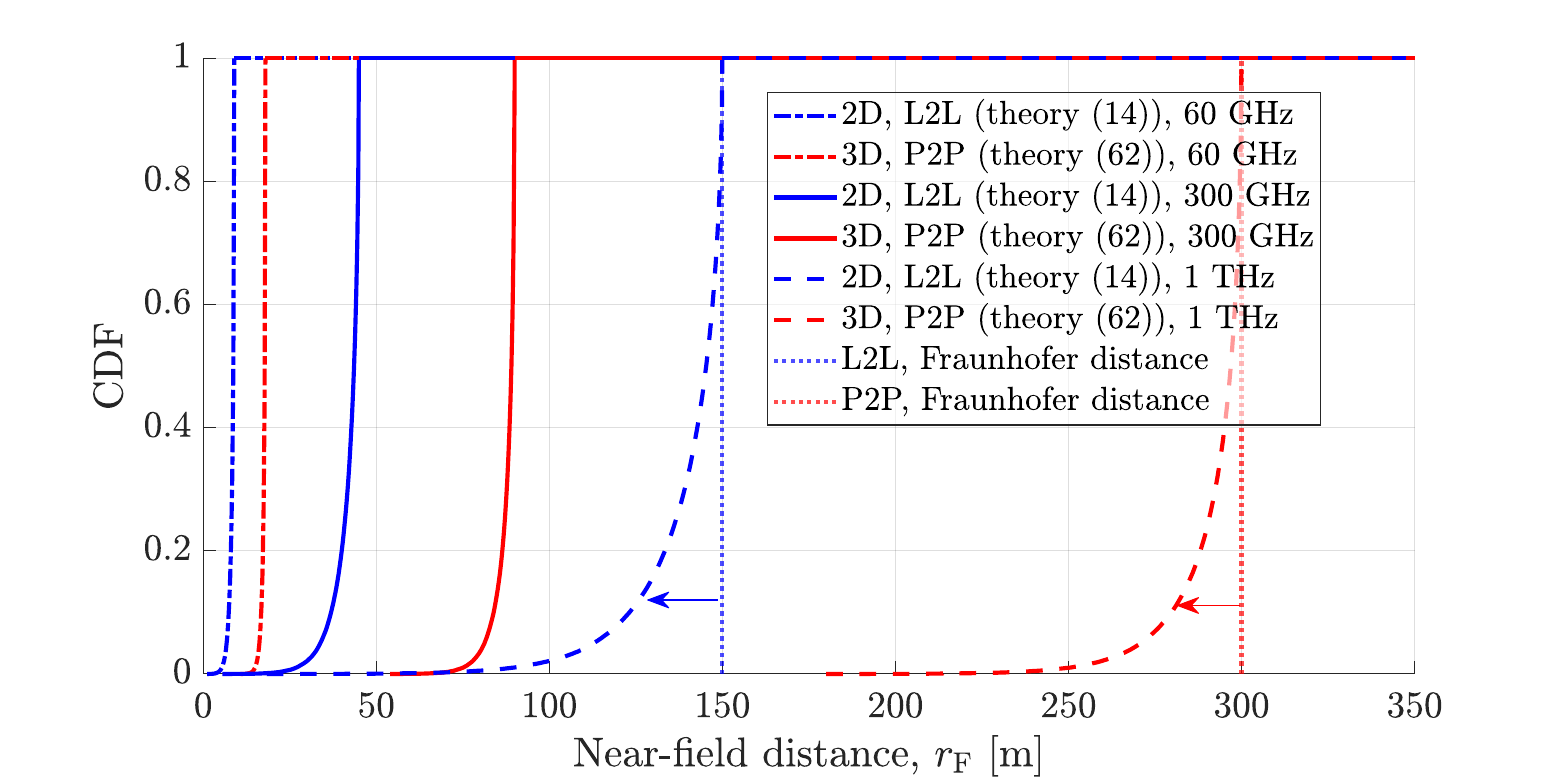}
        \caption{CDF of the near-field distance $r_\mathrm{F}$.}
        \label{fig:CDF_r_F}
    \end{subfigure}
    \hspace{-9mm}
    \begin{subfigure}[t]{0.51\linewidth}
        \centering
        \includegraphics[width=\linewidth]{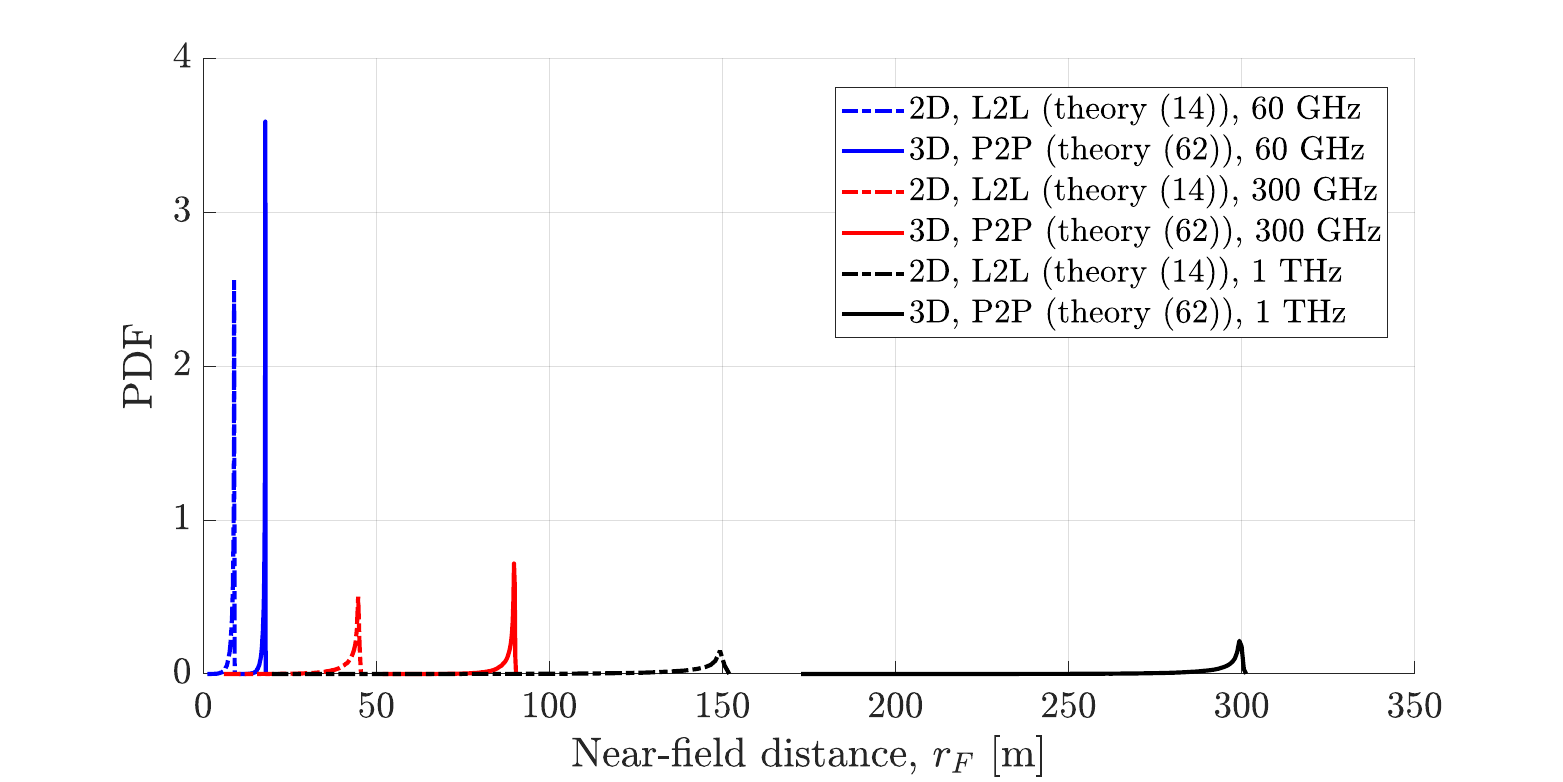}
        \caption{PDF of the near-field distance $r_\mathrm{F}$.}
        \label{fig:PDF_r_F}
    \end{subfigure}
    
    \caption{Statistical characterization of the near-field distance $r_\mathrm{F}$ with $D_1 = 0.1$\,m and $D_2 = 0.05$\,m.}
    \label{fig:CDF_PDF_r_F}
\end{figure*}

\subsection{Effect of misalignment}
\rew{To investigate the statistical characteristics of the near-field distance under angular uncertainty, we model random variations in $\theta$, $\phi$, and $\theta'$ using the truncated von Mises (TvM) distribution~\cite{glazunov2018spherical}. This choice is motivated by the fact that misalignment angles are circular random variables, for which the von Mises distribution provides a natural ``circular analogue'' of the Gaussian distribution with a single concentration parameter controlling the spread, and the truncation matches the bounded angle ranges in our system model. Similar von Mises-type models have also been adopted in related mmWave literature to characterize angular uncertainty and device-orientation statistics~\cite{glazunov2018spherical, nazari2023mmwave}.} Unless stated otherwise, we set the concentration parameter to $\kappa=10$ and the mean direction to $\mu=0$ (on-boresight $\alpha=0^\circ$). Fig.~\ref{fig:CDF_r_F} and Fig.~\ref{fig:PDF_r_F} respectively show the CDF and PDF of $r_{\mathrm{F}}$ for both 2D L2L and 3D P2P configurations with $\alpha = 0^\circ$. As shown in the CDF curves (Fig.~\ref{fig:CDF_r_F}), both 2D and 3D configurations exhibit sharply concentrated distributions around their nominal near-field distance values. However, the 3D P2P scenario displays significantly larger average values, especially at 1\,THz. \rew{For reference, Fig.~\ref{fig:CDF_r_F} also includes two vertical lines indicating the conventional Fraunhofer distances at $1$\,THz. The arrows highlight the resulting mismatch: under random rotations, the distribution of $r_{\mathrm F}$ can be noticeably shifted away from this no-rotation baseline.} Fig.~\ref{fig:PDF_r_F} provides the corresponding PDF view, where 3D P2P exhibits right-shifted distributions compared to 2D L2L. These results indicate that random misalignment can introduce non-negligible variability in $r_{\mathrm F}$ and can also shift its typical value away from the conventional no-rotation Fraunhofer baseline.

\begin{figure*}[!t]
\centering
\subfloat[$\alpha = 0^{\circ}$]{\includegraphics[width=0.34\textwidth]{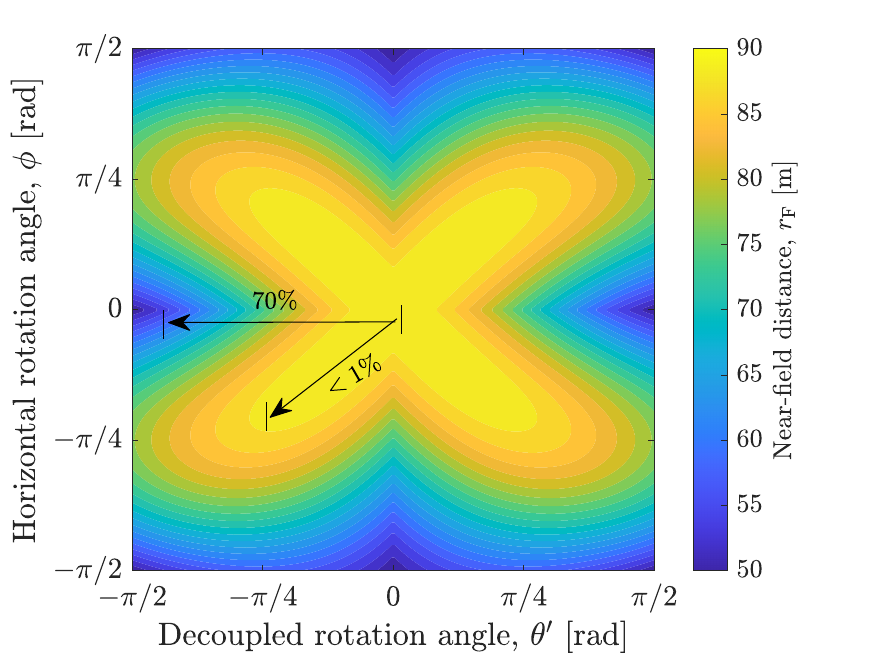}%
\label{fig:nearfield_beta_0}}
\hspace{-0.3cm}
\subfloat[$\alpha = 30^{\circ}$]{\includegraphics[width=0.34\textwidth]{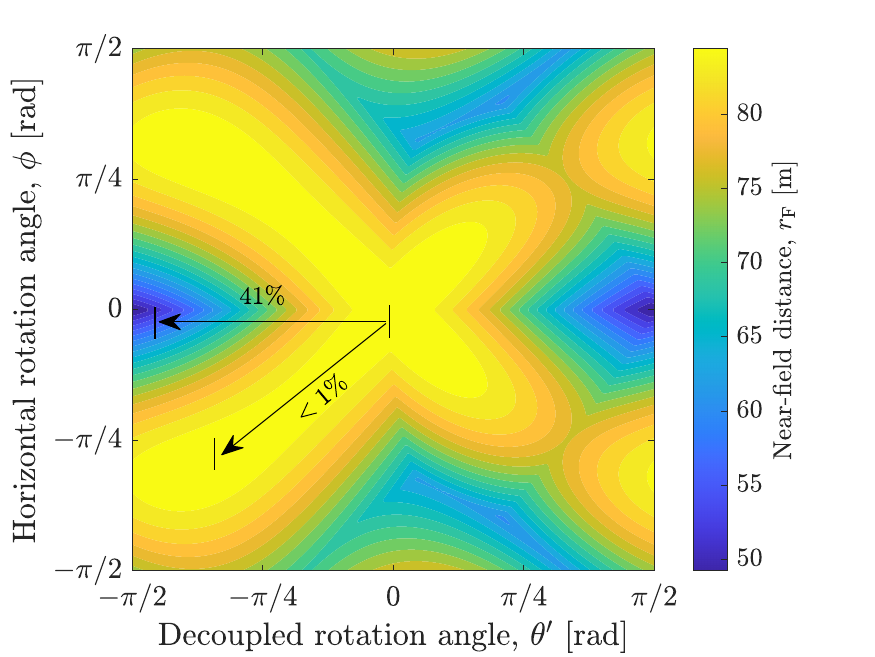}%
\label{fig:nearfield_beta_30}}
\hspace{-0.3cm}
\subfloat[$\alpha = 60^{\circ}$]{\includegraphics[width=0.34\textwidth]{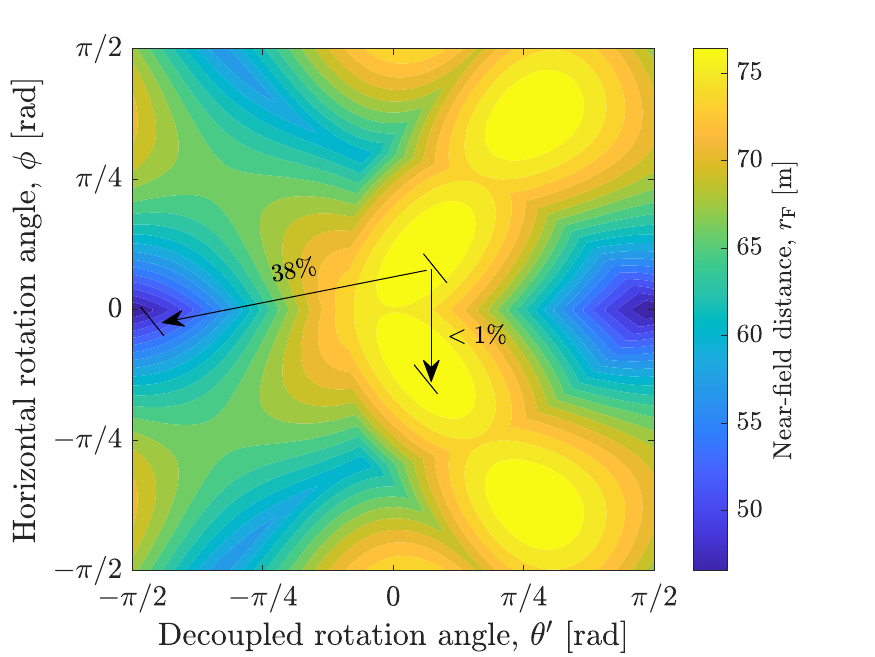}%
\label{fig:nearfield_beta_90}}
\caption{Near-field distance distribution $r_{\mathrm{F}}$ in the $(\theta', \phi)$ parameter space for different off-boresight angles $\alpha$.}
\label{fig:nearfield_distribution_alpha}
\end{figure*}

To further understand how joint azimuth and elevation misalignment affects $r_{\mathrm{F}}$, Fig.~\ref{fig:nearfield_distribution_alpha} presents $r_{\mathrm{F}}$ over the $(\theta',\phi)$ rotation space for $\alpha = 0^\circ$, $30^\circ$, and $90^\circ$, under the 3D P2P configuration. Under on-boresight conditions ($\alpha = 0^\circ$) in Fig.~\ref{fig:nearfield_beta_0}, $r_{\mathrm{F}}$ exhibits a strong dependence on the horizontal rotation angle $\phi$. \rew{In addition to the overall shape, we annotate in Fig.~\ref{fig:nearfield_distribution_alpha} illustrative relative deviations with respect to the no-rotation baseline, defined as $100\times\Delta_{r_{\mathrm F}}$, where $\Delta_{r_{\mathrm F}}\triangleq |r_{\mathrm F}(\cdot)-r_{\mathrm F}^{(0)}|/r_{\mathrm F}^{(0)}$ is defined above. Specifically, the percentages (e.g., $70\%$, $41\%$, and $38\%$) correspond to the relative deviation evaluated at the representative locations indicated by the arrows in each subfigure, which illustrates how ignoring rotations can lead to sizable discrepancies depending on $(\theta',\phi)$ and $\alpha$.} \textbf{These results highlight the intricate dependence of $r_{\mathrm{F}}$ on joint rotations, including asymmetry and angle-coupling effects that are not captured by conventional no-rotation baselines, which motivates the use of the proposed misalignment-aware expressions for accurate near-field characterization under general misalignment conditions.}

\section{Conclusion}\label{sec:conclusion}
This paper presents a generalized framework for characterizing the radiative near-field distance in mmWave and THz communications under antenna array misalignment. By accounting for various practical rotations of the Tx and Rx antenna arrays, as well as extending the classical phase-deviation criterion to a tunable threshold $\varphi \in (0,\pi]$, we derived \emph{closed-form exact and approximate expressions for a wide range of possible configurations}, including both array-to-array (ULA-to-ULA and UPA-to-UPA) setups and array-to-point special cases. The proposed expressions provide a more flexible and accurate representation of the near-field region boundaries in realistic mmWave/THz communication scenarios.

\rew{Our numerical results further quantify the impact of misalignment on the near-field boundary by comparing the misalignment-aware expressions with the corresponding no-rotation baselines. In the considered mmWave/THz settings and rotation ranges, ignoring rotations can lead to sizable deviations in near-field distance (e.g., up to $70\%$ in Fig.~\ref{fig:nearfield_distribution_alpha} and around $7\%$--$28\%$ for representative device-size combinations in Fig.~\ref{fig:diff_RxSize}), which may translate into shifts of several tens of~meters.} 
\rew{From an operational viewpoint, such shifts directly affect the distance beyond which conventional far-field beamsteering can be used with negligible gain loss, and where more complex range-dependent near-field beamfocusing/codebooks would otherwise be invoked.}

\rew{Consequently, the derived closed-form boundaries are directly useful for system evaluation and design: they enable identifying rotation-admissible regimes where simple distance-independent beamsteering is sufficient, as well as regimes where near-field handling is genuinely needed.} Hence, a \emph{careful consideration of spatial misalignment and array geometry (e.g., via the models developed in this article) becomes essential} for accurate evaluation and better design of forthcoming mmWave/THz links that operate across near-, far-, and cross-field regimes beyond~6G.

\appendices

\section{Proof of Proposition~\ref{prop:dominant_a}}
\label{appendix:proof_dominant_a}
Using the identity $|a - b| - |a + b| = -2 \cdot \operatorname{sgn}(ab) \cdot \min(|a|, |b|)$, we obtain
\begin{equation}
\begin{split}
    & r_{(\text{a})}^{\text{L2L}}(\theta,\alpha) -  r_{(\text{b})}^{\text{L2L}}(\theta,\alpha) = \frac{\pi D_1 D_2 \cos \theta' \cos \alpha}{\lambda \varphi} 
    \\&- \operatorname{sgn}(\sin \theta'\sin \alpha)\min( D_1|\sin \theta'|, D_2|\sin \alpha| ).
\end{split}
\end{equation}

When $\theta'$ and $\alpha$ have opposite signs (i.e., $\theta'\alpha<0$), such as $\theta' \in (-\pi/2, 0)$ and $\alpha \in (0, \pi/2)$ or $\theta' \in (0, \pi/2)$ and $\alpha \in (-\pi/2, 0)$, it holds that $\cos \theta' \cos \alpha > 0$ and $\sin \theta' \sin \alpha < 0$. Consequently, $\operatorname{sgn}(\sin \theta'  \sin \alpha) = -1$, and the difference becomes strictly positive $r_{\alpha ,(\mathrm{a})}^{\mathrm{U}2\mathrm{U}} > r_{\alpha ,(\mathrm{b})}^{\mathrm{U}2\mathrm{U}}$. Therefore, case~(a) always dominates when $\theta'$ and $\alpha$ have opposite signs.

Next, we consider the case where $\theta'$ and $\alpha$ have the same sign (i.e., $\theta'\alpha>0$), which implies $\operatorname{sgn}(\sin \theta' \sin \alpha) = 1$. In this case, a sufficient condition for 
$ r_{(\text{a})}^{\text{L2L}}(\theta,\alpha)>  r_{(\text{b})}^{\text{L2L}}(\theta,\alpha)$ is
\begin{equation}
    {\pi D_1 D_2 \cos \theta' \cos \alpha} > \lambda \varphi\min\left( D_1|\sin \theta'|,\, D_2|\sin \alpha| \right).
\end{equation}

To obtain a closed-form bound, we consider the symmetric case where $\theta' = \alpha = \vartheta$. In this case, the inequality reduces to ${\max(D_1, D_2) \pi \cos^2 \vartheta} > \lambda \varphi\sin \vartheta.$ Letting $\kappa= {\max(D_1, D_2)\pi}/{(\lambda \varphi)}$, we rearrange the inequality as the quadratic function $\kappa \sin^2 \vartheta + \sin \vartheta - \kappa < 0.$ By setting the left-hand side equal to zero, we obtain the threshold angle
\begin{equation}
    \vartheta_{\mathrm{th}} = \sin^{-1} \left( \frac{-1 + \sqrt{1 + 4\kappa^2}}{2\kappa} \right),
\end{equation}
below which case~(a) dominates. This concludes the proof.

\section{Comparison between UPA and ULA Expressions}
\label{appendix:UPA-ULA-comparison}
For a deeper comparison between UPA and ULA configurations, we consider a \emph{diagonal-matched} setting where the UPA aperture edge lengths are scaled such that its diagonal equals the ULA aperture length. Specifically, we set $D_1^{\mathrm{UPA}} = {D_1}/{\sqrt{2}}$ and $D_2^{\mathrm{UPA}} = {D_2}/{\sqrt{2}}$. Substituting these into~\eqref{equ:UPA-to-UPA_off-boresight_0} yields
\begin{equation}
    r_{\mathrm{F,off}}^{\mathrm{P2P}} = \frac{\pi (D_1 + D_2)^2}{8 \lambda \varphi} 
    + \frac{\pi (D_1 \cos(\theta - \alpha) + D_2 \cos \alpha)^2}{8 \lambda \varphi}.
    \label{equ:UPA_diag_matched}
\end{equation}

We compute the difference between~\eqref{equ:UPA_diag_matched} and~\eqref{equ:ULA-to-ULA_off_boresight} as
\begin{equation}
    \Delta r_{\mathrm{F,off}}= \frac{\pi \left( D_1\sin \left( \theta -\alpha \right) +D_2\sin \alpha \right) ^2}{8\lambda \varphi}+\Gamma \left( \theta ,\alpha \right),
\end{equation}
where $\Gamma \left( \theta ,\alpha \right) \triangleq {\pi D_1D_2\left( 1-\cos \left( 2\alpha -\theta \right) \right)}/({4\lambda \varphi}).$

Since $1-\cos(2\alpha-\theta) \geq 0$ for all $\theta,\alpha$, both terms in $\Delta r_{\mathrm{F,off}}$ are nonnegative. Hence, $\Delta r_{\mathrm{F,off}} \geq 0$ in general, which indicates that the UPA always yields a larger near-field distance than the ULA. Physically, this nonnegative gap originates from the UPA’s two-dimensional aperture structure, which introduces an additional spatial degree of freedom for phase accumulation.  Even after diagonal matching, the orthogonal aperture dimension contributes extra geometric phase errors that enlarge the near-field boundary compared to the purely one-dimensional ULA case.

\bibliographystyle{IEEEtran}
\bibliography{references}

\begin{IEEEbiography}[{\includegraphics[width=1in,height=1.25in,clip,keepaspectratio]{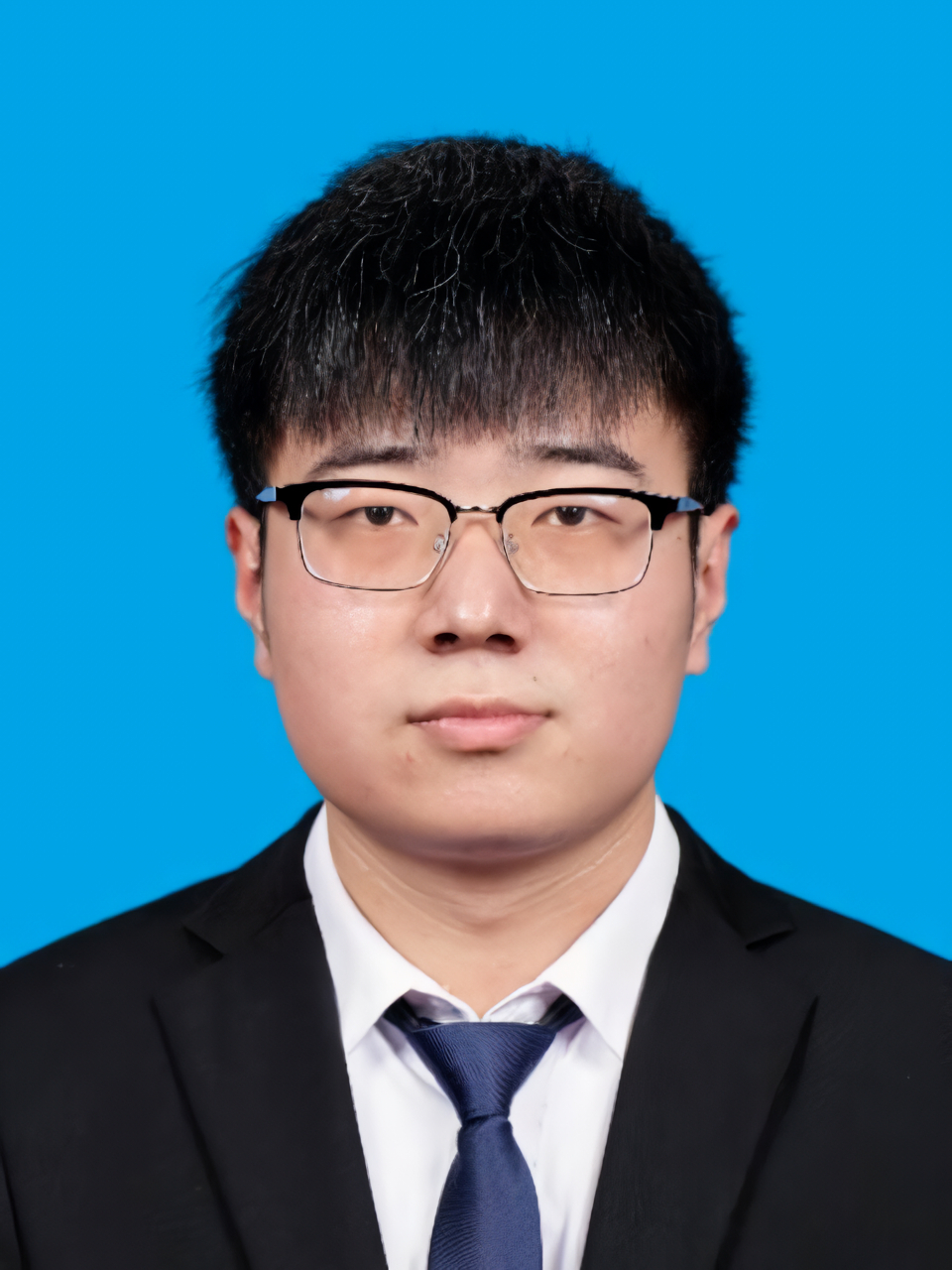}}]{Peng Zhang}
received the B.S. degree in communication engineering from Beijing Jiaotong University, Beijing, China, in 2021, and the M.S. degree in information and communication engineering from the same university in 2024. He is currently pursuing the Ph.D. degree in information and communications technology with KTH Royal Institute of Technology, Stockholm, Sweden. His research interests include near-field wireless communications, sub-terahertz (sub-THz) and THz systems.
\end{IEEEbiography}

\begin{IEEEbiography}[{\includegraphics[width=1in,height=1.25in,clip,keepaspectratio]{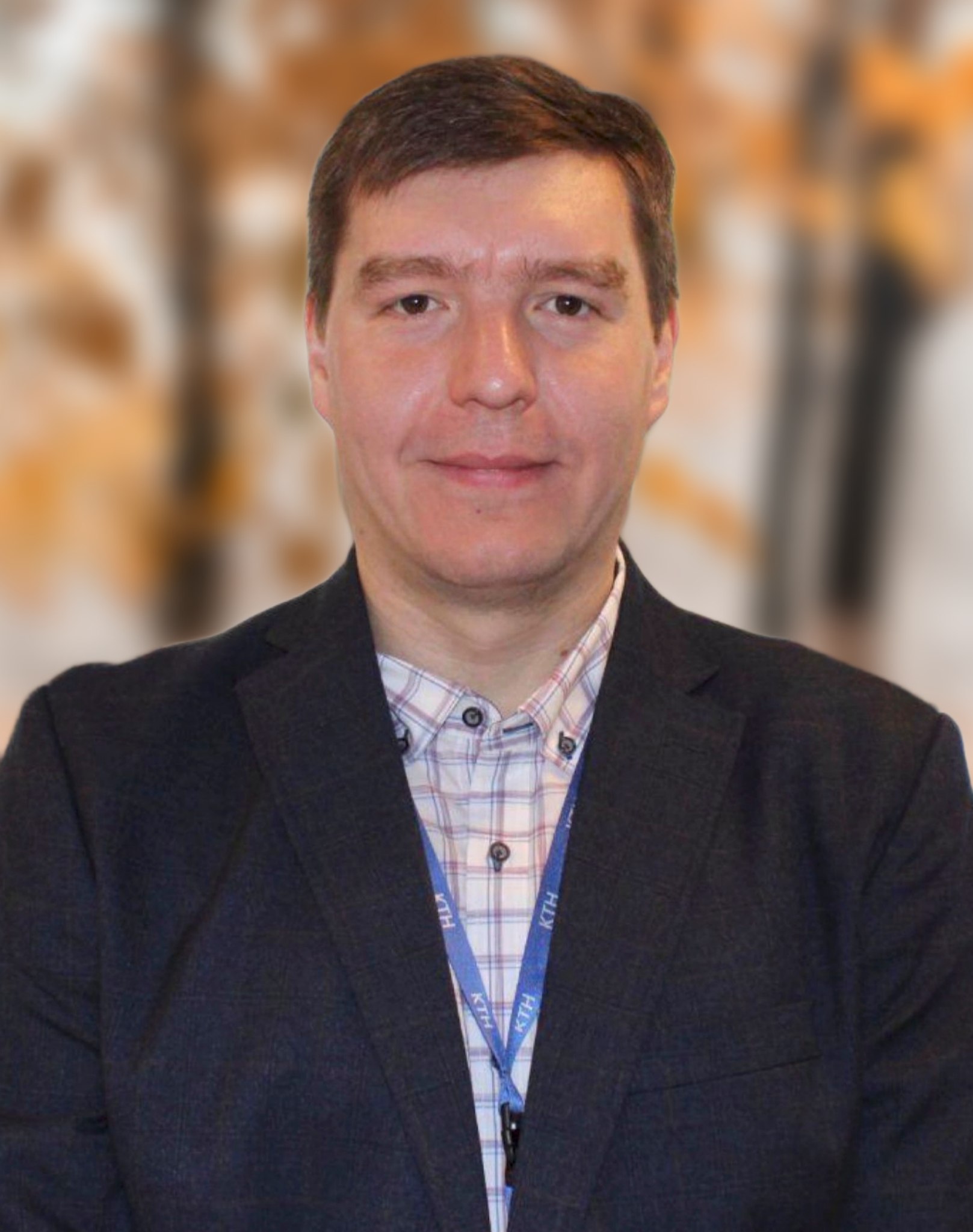}}]{Vitaly Petrov (M'20)}
is an Assistant Professor with the Department of Communication Systems, KTH Royal Institute of Technology, Stockholm, Sweden. Prior to joining KTH in 2024, he was a Principal Research Scientist with Northeastern University, Boston, MA, USA, from 2022 to 2024, and a Senior Standardization Specialist and 3GPP RAN1 Delegate with Nokia Bell Labs and later Nokia Standards from 2020 to 2022. He received the Ph.D. degree in communications engineering from Tampere University, Finland, in 2020. He has also been a Visiting Researcher with The University of Texas at Austin, Georgia Institute of Technology, and King's College London. His current research interests include terahertz-band communications and networking.
\end{IEEEbiography}

\begin{IEEEbiography}[{\includegraphics[width=1in,height=1.25in,clip,keepaspectratio]{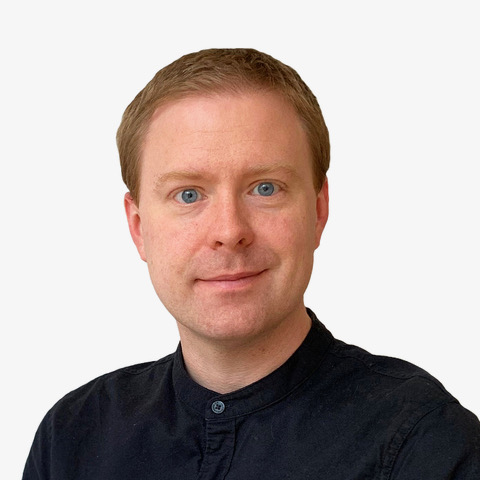}}]{Emil Bj\"ornson (Fellow, IEEE)} received the M.S. degree in engineering mathematics from Lund University, Sweden, in 2007, and the Ph.D. degree in telecommunications from the KTH Royal Institute of Technology, Sweden, in 2011.

From 2012 to 2014, he was a Post-Doctoral Researcher with the Alcatel-Lucent Chair on Flexible Radio, SUPELEC, France. From 2014 to 2021, he held different professor positions at Link\"oping University, Sweden. He has been a Full Professor of Wireless Communication at KTH since 2020 and currently leads the Department of Communication Systems. He has authored the textbooks \emph{Optimal Resource Allocation in Coordinated Multi-Cell Systems} (2013), \emph{Massive MIMO Networks: Spectral, Energy, and Hardware Efficiency} (2017), \emph{Foundations of User-Centric Cell-Free Massive MIMO} (2021), and \emph{Introduction to Multiple Antenna Communications and Reconfigurable Surfaces} (2024). He is dedicated to reproducible research and has published much simulation code. He researches multi-antenna communications, reconfigurable intelligent surfaces, radio resource allocation, machine learning for communications, and energy efficiency.

Dr. Bj\"ornson has performed MIMO research since 2006. His papers have received more than 40000 citations, he has filed more than 30 patent applications, and he is recognized as a Clarivate Highly Cited Researcher. He co-hosts the podcast Wireless Future and has a popular YouTube channel with the same name. He is a Wallenberg Academy Fellow, a Digital Futures Fellow, and an SSF Future Research Leader. He has received the 2014 Outstanding Young Researcher Award from IEEE ComSoc EMEA, the 2015 Ingvar Carlsson Award, the 2016 Best Ph.D. Award from EURASIP, the 2018 and 2022 IEEE Marconi Prize Paper Awards in Wireless Communications, the 2019 EURASIP Early Career Award, the 2019 IEEE ComSoc Fred W. Ellersick Prize, the 2019 IEEE Signal Processing Magazine Best Column Award, the 2020 Pierre-Simon Laplace Early Career Technical Achievement Award, the 2020 CTTC Early Achievement Award, the 2021 IEEE ComSoc RCC Early Achievement Award, the 2023 IEEE ComSoc Outstanding Paper Award, and the 2024 IEEE ComSoc Stephen O. Rice Prize. He also coauthored papers that received best paper awards at six international conferences.
\end{IEEEbiography}

\end{document}